\documentclass[journal]{IEEEtran}

\usepackage{amsmath}
\usepackage{amsfonts}
\usepackage{amssymb}
\usepackage{amsthm}
\usepackage{tabularx}
\usepackage{mathrsfs} 
\usepackage{array}
\usepackage{pifont}  % for \ding{55}

\usepackage{threeparttable}

\newcolumntype{C}[1]{>{\centering\arraybackslash}p{#1}}

\usepackage{url}
\usepackage[colorlinks]{hyperref}
\newtheorem{theorem}{Theorem}

\newtheorem{lemma}{Lemma}
\newtheorem{corollary}{Corollary}

\usepackage{stfloats}
\usepackage{float}
\usepackage{graphicx}
\usepackage{xcolor}
\usepackage{subfigure}

\makeatletter
\def\blfootnote{\xdef\@thefnmark{}\@footnotetext}
\makeatother

\makeatletter
\newcommand*{\rom}[1]{\expandafter\@slowromancap\romannumeral #1@}
\makeatother
\begin{document}
	
%\title{On Performance of FAS-assisted NOMA-ISAC:\\Backscatter Model and Analysis}
\title{Performance Analysis of FAS-Aided NOMA-ISAC:\\A Backscattering Scenario}  
\author{Farshad~Rostami~Ghadi,~\IEEEmembership{Member}, \textit{IEEE}, 
             Kai-Kit~Wong,~\IEEEmembership{Fellow}, \textit{IEEE},
             F. Javier~Lopez-Martinez,~\IEEEmembership{Senior Member}, \textit{IEEE}, 
             Hyundong Shin,~\IEEEmembership{Fellow}, \textit{IEEE}, and Lajos Hanzo,~\IEEEmembership{Life Fellow}, \textit{IEEE}
\vspace{-9mm}
}

\maketitle

\begin{abstract}
This paper investigates a two-user downlink system for integrated sensing and communication (ISAC) in which the two users deploy a fluid antenna system (FAS) and adopt the non-orthogonal multiple access (NOMA) strategy. Specifically, the integrated sensing and backscatter communication (ISABC) model is considered, where a dual-functional base station (BS) serves to communicate the two users and sense a tag's surrounding. In contrast to conventional ISAC, the backscattering tag reflects the signals transmitted by the BS to the NOMA users and enhances their communication performance. Furthermore, the BS extracts environmental information from the same backscatter signal in the sensing stage. Firstly, we derive closed-form expressions for both the cumulative distribution function (CDF) and probability density function (PDF) of the equivalent channel at the users utilizing the moment matching method and the Gaussian copula. Then in the communication stage, we obtain closed-form expressions for both the outage probability and for the corresponding asymptotic expressions in the high signal-to-noise ratio (SNR) regime. Moreover, using numerical integration techniques such as the  Gauss-Laguerre quadrature (GLQ), we have series-form expressions for the user ergodic communication rates (ECRs). In addition, we get a closed-form expression for the ergodic sensing rate (ESR) using the Cram\'er-Rao lower bound (CRLB). Finally, the accuracy of our analytical results is validated numerically, and we confirm the superiority of employing FAS over traditional fixed-position antenna systems in both ISAC and ISABC. 
\end{abstract}

%\vspace{-3.5ex}

\maketitle
%\blfootnote{\noindent Copyright (c) 2015 IEEE. Personal use of this material is permitted. However, permission to use this material for any other purposes must be obtained from the IEEE by sending a request to pubs-permissions@ieee.org.} 
%\blfootnote{Manuscript received January 25, 2021; revised XXX. The review of this paper was coordinated by XXXX.} 
\blfootnote{The work of F. Rostami Ghadi and K. K. Wong is supported by the Engineering and Physical Sciences Research Council (EPSRC) under Grant EP/W026813/1. The work of F. J. L\'opez-Mart\'inez  is funded in part by Junta de Andaluc\'ia through grant EMERGIA20-00297, and in part by MICIU/AEI/10.13039/50110001103 through grant PID2023-149975OB-I00 (COSTUME).}
\blfootnote{\noindent F. Rostami Ghadi and K. K. Wong are with the Department of Electronic and Electrical Engineering, University College London, London, UK. K. K. Wong is also affiliated with the Department of Electronic Engineering, Kyung Hee University, Yongin-si, Gyeonggi-do 17104, Korea. (e-mail:$\{\rm f.rostamighadi,kai\text{-}kit.wong\}@ucl.ac.uk$).}
\blfootnote{\noindent F. J. L\'opez-Mart\'inez is with the Department of Signal Theory, Networking and Communications, Research Centre for Information and Communication Technologies (CITIC-UGR), University of Granada, 18071, Granada (Spain), and also with the Communications and Signal Processing Lab, Telecommunication Research Institute (TELMA), Universidad de M\'alaga, M\'alaga, 29010, (Spain). (e-mail: $\rm fjlm@ugr.es$).}
%\blfootnote{\noindent C. B. Chae is with School of Integrated Technology, Yonsei University, Seoul, 03722, Korea. (e-mail: $\rm cbchae@yonsei.ac.kr$).}
\blfootnote{\noindent H. Shin is with the Department of Electronic Engineering, Kyung Hee University, Yongin-si, Gyeonggi-do 17104, Korea. (e-mail:$\{\rm hshin@khu.ac.kr$).}
\blfootnote{\noindent L. Hanzo is with the School of Electronics and Computer Science, University of Southampton, Southampton, U.K. (e-mail: $\rm lh@ecs.soton.ac.uk$).}
	
%\blfootnote{Digital Object Identifier 10.1109/XXX.2021.XXXXXXX}
\blfootnote{Corresponding authors: Kai-Kit Wong, Hyundong Shin.}

%\IEEEpeerreviewmaketitle

\section{Introduction}\label{sec-intro}
\IEEEPARstart{T}{he rapid} proliferation of intelligent devices and the fast escalating demand for high-efficiency wireless communication present critical challenges in the next-generation of wireless networks, a.k.a.~the sixth generation (6G) \cite{Fariq-2020,You-2023}. One growing technology, which has become a key use case in 6G, is integrated sensing and communication (ISAC) \cite{liu2022integrated}. ISAC merges sensing and communication (S\&C) functionalities into a single framework. Different from conventional frequency-division sensing and communication techniques which require separate frequency bands and dedicated hardware infrastructures for joint S\&C, ISAC provides a more efficient solution in terms of spectrum, energy, and hardware usage \cite{zhang2021overview,zhang2021enabling}.

Although ISAC is appealing, it makes the provision of S\&C functionalities harder, not easier as we are required to squeeze more from what is already limited in the physical layer. As a result, in this context, there is a pressing desire to increase the degrees-of-freedom (DoF) in the physical layer. For decades, multiple-input multiple-output (MIMO) has been responsible for raising the DoF for great benefits if multiple antennas are deployed at both ends. MIMO however comes with expensive radio-frequency (RF) chains and is also subject to strict space constraints. Recently, a promising solution to overcome these issues is to adopt the fluid antenna system (FAS) technology as opposed to a traditional antenna system (TAS) \cite{Wong-2020cl,wong2020fluid}. A FAS represents the new form of reconfigurable antennas that enables shape and position flexibility \cite{wong2022bruce}. With FAS, a wireless communication channel can deliver more diversity with less space and less number of RF chains \cite{Khammassi-2023,Vega-2023,Vega-2023-2,Alvim-2023,Psomas-dec2023}. Recent efforts also studied the use of FAS at both ends of the channel \cite{New-twc2023}. Channel estimation for FAS has also become an important research problem \cite{Hao-2024,Dai-2023,Zou-2023}. In addition, artificial intelligence (AI) techniques are increasingly relevant to the design and optimization of FAS \cite{Wang-aifas2024,Waqar-2024}. The emerging movable antenna systems also fall under the category of FAS, specifying the implementation of using stepper motors \cite{Zhu-Wong-2024}. Encouraging experimental results have recently been reported to validate the promising performance of FAS \cite{Shen-tap_submit2024,Zhang-pFAS2024}.

On the other hand, backscatter communication (BC) is rising as a cost-effective technique for ultra-low-power communication technologies \cite{lu2018ambient}. BC exploits passive reflection of existing ambient RF signals to transmit data, thereby eliminating the need for an internal power source for transmission. Hence, BC has great potential to significantly enhance the efficiency and sustainability of wireless networks, facilitating widespread connectivity among a diverse array of Internet-of-Things (IoT) devices and sensors in 6G networks \cite{akyildiz20206g}. Overall, it is clear that ISAC, FAS and BC can be strategically combined for an efficient solution, which motivates the work of this paper.

\subsection{State-of-the-Art}
Some efforts have been undertaken in the context of ISAC, FAS, and BC from various aspects and some of their intersections in recent years. For example, \cite{ouyang2022performance1} analyzed the diversity order to evaluate the sensing rate (SR) and communication rate (CR) for both downlink and uplink ISAC systems. Also, \cite{ouyang2022performance} derived novel expressions for the outage probability (OP), the ergodic CR (ECR), and the SR for ISAC. By extending the results in \cite{ouyang2022performance1,ouyang2022performance} to a downlink MIMO system for ISAC, the diversity orders and high signal-to-noise ratio (SNR) slopes of the SR and CR were subsequently derived in \cite{ouyang2023mimo}. 

Besides, non-orthogonal multiple access (NOMA) was often considered in the application of ISAC because it expands the capacity region by making better use of the available DoF. Recently in \cite{liu2022performance}, the authors studied a cooperative ISAC network for non-orthogonal downlink transmission, where they characterized the exact and asymptotic OP, the ECR, and the probability of successful sensing detection. Also, the sum of ECR and the signal-to-interference plus noise ratio (SINR) of the sensing signal were maximized. In \cite{wang2022noma}, a dual-functional base station (BS) was considered to serve users in the uplink using NOMA and a beamforming design problem was studied to maximize the weighted sum of CR and the effective sensing power. Later in \cite{ouyang2023revealing}, the impact of successive interference cancellation (SIC) in NOMA-ISAC was investigated. Moreover, using the concept of semi-ISAC under NOMA, \cite{zhang2023semi} obtained analytical expressions for the OP, the ECR, and the ergodic radar estimation information rate. Most recently, the CR and SR for a near-field ISAC system was studied in \cite{zhao2024modeling}. 

On the other hand, there have also been significant research on FAS as mentioned earlier. Some recent efforts have looked into the combination of FAS and NOMA for communication \cite{new2023fluidnoma,ghadi2024performancewcn} with the latter work also considering wireless power transfer. ISAC using FAS is not well understood nonetheless, with only few recent work beginning to highlight the benefits of FAS for ISAC scenarios. For example, with FAS at the BS serving in the downlink for ISAC, the authors in \cite{zhang2024efficient} devised a proximal distance algorithm to solve the multiuser sum-rate maximization problem with a radar sensing constraint to obtain the closed-form beamforming vector, and also an extrapolated projected gradient algorithm to obtain a better antenna location configuration for FAS to enhance the ISAC performance. Most recently, machine learning techniques have also been used to optimize a downlink MIMO network with FAS at the BS for ISAC subjected to a sensing constraint \cite{Wang-fasisac2023}. Further, masked autoencoders were employed to only exploit partial channel state information (CSI) at the BS with great effect.

Comparatively, the intersection between BC and other technologies is even less explored. The only work considering both FAS and BC appears to be \cite{ghadi2024performancebc}, in which a novel closed-form expression for the cumulative distribution function (CDF) of the equivalent channel at the FAS-enabled reader was obtained and the impact of FAS was highlighted. On the other hand, the concept of integrated sensing and backscatter communication (ISABC) was proposed in \cite{galappaththige2023integrated}. Conventional ISAC integrates S\&C using active devices and passive objects, while ISABC enhances this capability by incorporating backscatter tags. The tags enable passive objects equipped with them to reflect and modulate existing signals from the ISAC system. In summary, however, FAS-aided ISABC has not been studied before. The unique contributions of our work are prominently highlighted in Table \ref{table1}, allowing for an easy comparison with existing studies. The detailed contributions of this paper are discussed in the following section.

\begin{table}\caption{Comparison of our contributions to the literature}\label{table1} \centering
	\begin{tabular}{ |C{1.2cm}|C{0.6cm}|C{0.6cm}|C{0.6cm}|C{0.9cm}|C{0.4cm}| C{0.4cm}|C{0.4cm}|  }
		\hline
		\textbf{Works} & \textbf{ISAC} & \textbf{FAS} & \textbf{BC} & \textbf{NOMA} & \textbf{SR} & \textbf{CR} & \textbf{OP}\\
		\hline
		\hspace{0.0cm}\cite{ouyang2022performance}, \cite{liu2022performance}, \hspace{-0.2cm}\cite{ouyang2023revealing}, \cite{zhang2023semi}& \vspace{0.01mm}\checkmark & \vspace{0.01mm} & \vspace{0.01mm} & \vspace{0.01mm}\checkmark & \vspace{0.01mm}\checkmark & \vspace{0.01mm}\checkmark & \vspace{0.01mm}\checkmark\\
		\hline
		\cite{zhao2024modeling}& \checkmark &  &  & & \checkmark&\checkmark&\\
		\hline
		\cite{new2023fluidnoma} &  & \checkmark &  & \checkmark&  & \checkmark &  \\
		\hline
		\cite{ghadi2024performancewcn} &  & \checkmark &  & &  &  & \checkmark\\
		\hline
		\cite{zhang2024efficient} & \checkmark & \checkmark &  &  & \checkmark & \checkmark &\\
		\hline
		\cite{Wang-fasisac2023} & \checkmark & \checkmark &  &  &  & \checkmark & \\
		\hline
		\cite{ghadi2024performancebc} &  & \checkmark & \checkmark & &  &  & \checkmark \\
		\hline
		\cite{galappaththige2023integrated} & \checkmark &  & \checkmark &  & \checkmark & \checkmark &  \\
		\hline
		Proposed & \checkmark & \checkmark & \checkmark & \checkmark & \checkmark & \checkmark & \checkmark\\
		\hline
	\end{tabular}
	\begin{tablenotes}
	\item \textbf{Table Notation.} The proposed work uniquely integrates all listed features (ISAC, FAS, BC, NOMA, SR, CR, and OP), providing a comprehensive approach not found in any of the other reviewed works. %Item is supported, $\times$: Item is not supported, $\otimes$: Item is conditionally supported (e.g., Optimization). 
	\end{tablenotes}
\end{table}
\subsection{Motivation and Contributions}
The beauty of ISAC for simultaneous S\&C with lessened hardwares and shared resources, the new DoF enabled by FAS and the great potential of BC in utilizing ambient signals, have motivated us to investigate the synergy between ISAC, FAS, and BC. The synergy of these technologies promises a robust, adaptable, and efficient communication framework ideal for supporting the demanding requirements of next-generation IoT and performance-critical wireless networks. Motivated by the above, we consider an ISABC scenario adopting the NOMA technology, where a backscatter tag serves as the sensing target and the communication users are equipped with FAS. The tag not only provides sensing information to the BS but transmits additional data to the FAS-equipped NOMA users. 

Specifically, our main contributions are as follows:
\begin{itemize}
\item First of all, we characterize both the CDF and probability density function (PDF) of the ISABC equivalent channel at the FAS-equipped NOMA users, exploiting the moment matching technique and copula theory. 
\item Then, in the communication stage, we derive the closed-form expressions for the OP, as well as their respective asymptotic expressions in the high SNR regime. Besides, we provide integral-form expressions of the ECR and approximate them using numerical integration techniques such as the Gauss-Laguerre quadrature (GLQ). 
\item Afterwards, we obtain a closed-form expression for the ergodic SR (ESR) in the sensing stage, using the Cram\'er-Rao lower bound (CRLB).
\item Finally, we evaluate the accuracy of our theoretical results and compare them to several benchmarks. Our findings indicate enhanced overall performance, demonstrating the superiority of FAS over TAS in both ISAC and ISABC scenarios. Additionally, we provide an accurate trade-off between SR and CR to gain further insights into the performance of the proposed system.
\end{itemize}

\subsection{Paper Organization}
The remainder of this paper is organized as follows. Section \ref{sec-sys} introduces the system model of ISABC for both the S\&C stages. The statistical characteristics of the equivalent channel, including communication performance such as OP and ECR analysis, and sensing performance such as ESR analysis, are presented in Section \ref{sys-perf}. Section \ref{sec-num} provides the numerical results, and finally, Section \ref{sec-con} concludes this paper.

\subsection{Mathematical Notations}
Throughout, we use boldface upper and lower case letters $\mathbf{X}$ and $\mathbf{x}$ for matrices and column vectors, respectively. $\mathbb{E}\left[\cdot\right]$ and $\mathrm{Var\left[\cdot\right]}$ denote the mean and variance operators, respectively. Moreover, $\left(\cdot\right)^T$,  $\left(\cdot\right)^{-1}$, $\left|\cdot\right|$, and $\mathrm{det}\left(\cdot\right)$ stand for the transpose, inverse, magnitude, and determinant, respectively.

\begin{figure*}[!t]
\centering
\includegraphics[width=1.3\columnwidth]{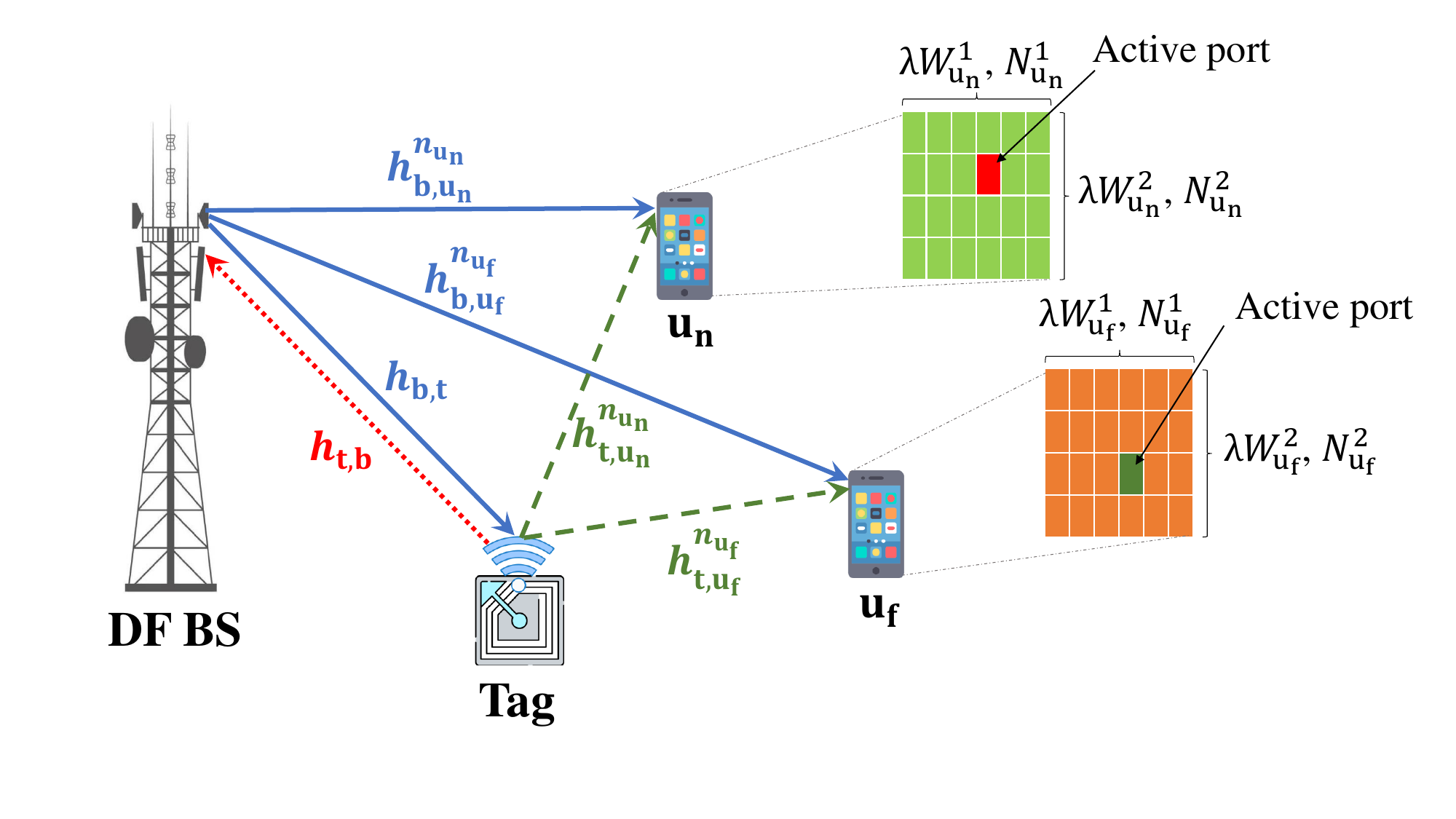}
\caption{The system model of ISABC with two FAS-aided users and one backscatter tag.}\label{fig-model}
\vspace{-2mm}
\end{figure*}
\section{System and Channel Models}\label{sec-sys}
\subsection{System Model}
As shown in Fig.~\ref{fig-model}, we consider a wireless communication system, where a dual-functional BS communicates with a pair of NOMA users, i.e., one near (strong) user $\mathrm{u_n}$ and one far (weak) user $\mathrm{u_f}$, while simultaneously broadcasts a sensing waveform to sense a single fixed-position antenna tag. Both users $\mathrm{u_n}$ and $\mathrm{u_f}$ are equipped with a planar FAS, while the BS includes two fixed-position antennas, where one is for transmitting and the other is for receiving signals. In this regard, we assume that the BS antennas are spatially well separated in order to cancel the self-interference. The FAS-equipped users  $i\in\left\{\mathrm{u_n,u_f}\right\}$, consist of a grid structure with $N^l_i$ ports that are uniformly distributed along a linear space of length $W^l_i\lambda$ for $l\in\left\{1,2\right\}$, i.e., $N_i=N^1_i\times N^2_i$ and $W_i=W^1_i\lambda\times W^2_i\lambda$. Additionally, a mapping function $\mathcal{F}\left(n_i\right)=\left(n_i^1,n_i^2\right)$, $n_i=\mathcal{F}^{-1}\left(n_i^1,n_i^2\right)$ is used to conveniently transform the two-dimensional (2D) indices to a one-dimensional (1D) index, so $n_i\in\left\{1,\dots,N_i\right\}$ and $n_i^l\in\left\{1,\dots,N_i^l\right\}$. Under this model, the BS transmits a signal $x(t)$ at the $t$-th time slot for both communication and sensing.\footnote{From a sensing perspective, $x(t)$ represents the radar snapshot transmitted at the $t$-th time slot, and for communications, it is the $t$-th data symbol \cite{ouyang2022performance1}.} The total power has to be split into two parts according to $\mu_\chi\in\left[0,1\right]$ for $\chi\in\left\{\mathrm{c,s}\right\}$, with one part for sensing and the other for communications, i.e., $\mu_\mathrm{c}+\mu_\mathrm{s} = 1$, in which the subscripts `$\mathrm{c}$' and `$\mathrm{s}$' specify the splitting factor for communications and sensing, respectively.

\subsection{Signal and Channel Models}
\subsubsection{Communication stage} 
The received signal at the $k_i$-th port of user $i$ via the tag at the $t$-th time slot is given by
\begin{align}\label{eq-yc}
y_{i}^{n_i}(t) = \sqrt{ P_\mathrm{b}\mu_\mathrm{c}}h^{n_i}_{\mathrm{eq},i}\left[\sqrt{p_\mathrm{u_n}}x_\mathrm{u_n}(t)+\sqrt{p_\mathrm{u_f}}x_\mathrm{u_f}(t)\right]+z_{i}^{n_i}(t),
\end{align}
in which $P_\mathrm{b}$ denotes the transmit power and $p_i$ is the power allocation factor satisfying $p_\mathrm{u_n}+p_\mathrm{u_f}=1$. Moreover, $x_{i}(t)$ denotes the symbol for user $i$ with $\mathbb{E}\left[|x_{i}(t)|^2\right]=1$. The term   $h^{n_i}_{\mathrm{eq},i}=d_{\mathrm{b},i}^{-\alpha} h_{\mathrm{b},i}^{n_i}+\zeta d_{\mathrm{b,t}}^{-\alpha}d_{\mathrm{t},i}^{-\alpha} h_\mathrm{b,t}h^{n_i}_{\mathrm{t},i}$ represents the equivalent channel (i.e., the sum of the direct and  effective backscatter channels) at the $n_i$-th port of user $i$, where $\zeta\in\left[0,1\right]$ is the tag's reflection coefficient,  $h^{n_i}_{\mathrm{b},i}$ defines the channel coefficient between the BS and the $n_i$-th port at the FAS-equipped user $i$, 
 $h_\mathrm{b,t}$ represents the channel coefficient between the BS and the tag, and  $h^{n_i}_{\mathrm{t},i}$ is the channel coefficient between the tag and the $n_i$-th port of user $i$. Additionally, $d_{\mathrm{b},i}$, $d_{\mathrm{b,t}}$, and $d_{\mathrm{t},i}$ represent, respectively, the distances from the BS to user $i$, the BS to the tag, and the tag to user $i$. Moreover, $\alpha>2$ denotes the path-loss exponent. Furthermore, $z_{i}^{n_i}(t)\sim\mathcal{CN}\left(0,\sigma^2\right)$ defines the additive white Gaussian noise (AWGN) with zero mean and variance $\sigma^2$ at the $n_i$-th port of user $i$. 
  %$\kappa_i=d_{\mathrm{b,\tilde{t}}_i}^{-\alpha}d_{\mathrm{\tilde{t},u}_i}^{-\alpha}$ defines the path-loss, where $d_{\mathrm{b,\tilde{t}}_i}$ is the distance between the BS and tag and $d_{\mathrm{\tilde{t},u}_i}$ is the distance between the tag and user $\mathrm{u}_i$.

\subsubsection{Sensing stage} Given that the backscattered signal at the tag is received by both users and the BS, the BS utilizes this signal to extract environmental information. In this context, we assume that the BS knows the type of pulse sent to the tag and has previously gathered observations to estimate the predicted range of the tag's position. If pulses are consistently transmitted to the tag at a fixed frequency, the BS can calculate the predicted echo using the prior observations. It is worth noting that uncertainty in the positioning is directly related to time delay fluctuations in radar systems \cite{chiriyath2015inner}. Therefore, under these assumptions, the time delay fluctuation $t_\mathrm{df}$ follows a Gaussian distribution with the variance of $\sigma^2_{t_\mathrm{df}}=\mathbb{E}\left[\left|t_\mathrm{df}-t_\mathrm{pre}\right|^2\right]$, in which $t_\mathrm{pre}$ is the predicted value of $t_\mathrm{df}$. As a result, the average power level of the echo signal, considering the uncertainty in the positioning decision, is formulated as \cite{chiriyath2015inner}
\begin{align}
\sigma^2_\mathrm{echo}=\mathbb{E}\left[\left|x\left(t-t_\mathrm{df}\right)-x\left(t-t_\mathrm{pre}\right)\right|^2\right]\approx\frac{\pi^2}{3}\sigma^2_{t_\mathrm{df}}.
\end{align}
As such, by transmitting the signal $x(t)$ to sense the tag surrounding environment, the BS receives the following reflected echo signal at its receive antenna, i.e., 
\begin{align}\label{eq-eco}
y_{\mathrm{b}}(t)=\sqrt{P_\mathrm{b}\mu_\mathrm{s}\zeta}h_{\mathrm{tsr,t}}\left[x\left(t-t_\mathrm{df}\right)-x\left(t-t_\mathrm{pre}\right)\right]+z_{\mathrm{b}}(t),
\end{align}
in which $z_\mathrm{b}(t)\sim\mathcal{CN}\left(0,\sigma^2\right)$ is the AWGN at the BS. Additionally, $h_{\mathrm{tsr,t}}=h_{\mathrm{t,b}}h_{\mathrm{b,t}}$ denotes the tag sensing response (TSR) that consists of both the channel from the BS to the tag and that from the tag to the BS. 

Furthermore, the ports at the FAS of each user can freely switch to a favourable position but they are in close proximity to each other, thus exhibiting spatial correlation in the channel coefficients. Assuming all links undergo Rayleigh fading channels, the covariance between two arbitrary ports $n_{i}$ and $\tilde{n}_{i}$ at each user $i$ in a three-dimensional (3D) environment with rich scattering can be defined as
\begin{align}
\varrho_i=j_0\left(2\pi\sqrt{\left(\frac{n^1_i-\tilde{n}^1_i}{N_i^1-1}W_i^1\right)^2+\left(\frac{n^2_i-\tilde{n}^2_i}{N_i^2-1}W_i^2\right)^2}\right),
\end{align}
where $\tilde{n}_{i}=\mathcal{F}^{-1}\left(\tilde{n}_{i}^1,\tilde{n}_{i}^2\right)$ for $\tilde{n}_{i}^l\in\left\{1,\dots,N_i^l\right\}$ and $j_0(.)$ is  the zero-order spherical Bessel function of the first kind.

\subsection{SINR Characterization}
Given that NOMA is used for the users, the optimal decoding technique involves implementing SIC according to the near FAS-equipped user $i$. To that end, the near user first decodes the signal transmitted to the far user treating its own signal as noise. %Then, SIC process is carried out and decoding of the near signal can be made interference-free.} 
Therefore, assuming that only the optimal port that maximizes the received SINR at the FAS-equipped users is activated, the SINR of the SIC process can be expressed as
\begin{align}\label{eq-sinr-sic}
\gamma_{\mathrm{sic}}=\frac{\overline{\gamma}p_\mathrm{u_f}\mu_\mathrm{c}\left|h_{\mathrm{eq,u_n}}^{n_\mathrm{u_n}^*}\right|^2}{\overline{\gamma}p_\mathrm{u_n}\mu_\mathrm{c}\left|h_{\mathrm{eq,u_n}}^{n_\mathrm{u_n}^*}\right|^2+1},
\end{align}
in which $\overline{\gamma}=\frac{P_\mathrm{b}}{\sigma^2}$ denotes the average SNR. Then, $\mathrm{u_n}$ removes the message of $\mathrm{u_f}$ from its received signal and decodes its own required information. As a consequence, the received SNR at the FAS-equipped user $\mathrm{u_n}$ is expressed as
%\begin{align}
%	\gamma_{\mathrm{sic},I}=\frac{\overline{\gamma}p_\mathrm{f}d_{\mathrm{b,u_n}}^{-\alpha}\left|h_{\mathrm{b,u_n}}^{k^*}\right|^2}{\overline{\gamma}p_\mathrm{n}d_{\mathrm{b,u_n}}^{-\alpha}\left|h_{\mathrm{b,u_n}}^{k^*}\right|^2+1}
%\end{align}
%\begin{align}
%	\gamma_{\mathrm{u_n},I}=\overline{\gamma}p_\mathrm{n}d_{\mathrm{b,u_n}}^{-\alpha}\left|h_{\mathrm{b,u_n}}^{k^*}\right|^2
%\end{align}
\begin{align}\label{eq-snr-un}
\gamma_{\mathrm{u_n}}=\overline{\gamma}p_\mathrm{u_n}\mu_\mathrm{c}\left|h_{\mathrm{eq,u_n}}^{n_\mathrm{u_n}^*}\right|^2.
\end{align}
Simultaneously, the far user $\mathrm{u_f}$ directly decodes its own signal but lacks the capability to filter out the signal from the near user $\mathrm{u_n}$ within the combined transmitted message. This implies that the far FAS-equipped user $\mathrm{u_f}$ must decode its signal directly by treating interference as noise. Thus, the received SINR at the far FAS-equipped user $\mathrm{u_f}$ is given by
\begin{align}
\gamma_{\mathrm{u_f}}=\frac{\overline{\gamma}p_\mathrm{u_f}\mu_\mathrm{c}\left|h_{\mathrm{eq,u_f}}^{n_\mathrm{u_f}^*}\right|^2}{\overline{\gamma}p_\mathrm{u_n}\mu_\mathrm{c}\left|h_{\mathrm{eq,u_f}}^{n_\mathrm{u_f}^*}\right|^2+1}.
\end{align}
Given the FAS at both users, we denote $n_i^*$ as the index of the best port at user $i$, given by
\begin{align}
n_i^*=\underset{n}{\arg\max}\left\{\left|h_{\mathrm{eq},i}^{n_i}\right|^2\right\}.
\end{align} 
Thus, the channel gain at the FAS-equipped user $i$ is given by
\begin{align}\label{eq-g-fas}
g_{\mathrm{fas},i}=\max\left\{\left|h_{\mathrm{eq},i}^1\right|^2,\dots,\left|h_{\mathrm{eq},i}^{N_i}\right|^2\right\},
\end{align}
in which $g_{\mathrm{eq},i}^{n_i}=\left|h_{\mathrm{eq},i}^{n_{i}}\right|^2$.

\section{Performance Analysis}\label{sys-perf}
Here, we first derive the CDF and PDF of the equivalent channel at the FAS-equipped NOMA users. Then, in the communication stage, we obtain compact analytical expressions for the OP and the corresponding asymptotic results, as well as derive the ECR. Additionally, we derive the closed-from expression for the ESR in the sensing stage. 

\subsection{Statistical Characteristics}
Given \eqref{eq-yc} and \eqref{eq-g-fas}, the equivalent channel at both users over the communication stage includes the maximum of $N_{i}$ correlated random variables (RVs), each of which includes the sum of two independent RVs, where one follows a Rayleigh distribution and the other one is the product of two independent exponentially-distributed RVs. In this regard, we exploit a flexible statistical approach in the statistics theory called Sklar's theorem, which can generate the joint CDF of many arbitrarily correlated RVs, e.g., $N_{i}$, beyond linear dependency by only using the marginal distribution of each RV and a specific function. For this purpose, let $\mathbf{s}=\left[S_1,\dots,S_d\right]$ be a vector of $d$ arbitrary correlated RVs having the univariate marginal CDF $F_{S_j}(s_j)$ and joint multivariate CDF $F_{S_1,\dots, S_d}\left(s_1,\dots,s_d\right)$ for $j\in\left\{1,\dots,d\right\}$, respectively. Then, in the extended real line domain $\overline{\mathbb{R}}$, the corresponding joint CDF is given by \cite{ghadi2020copula}
\begin{align}\label{eq-sklar}
F_{S_1,\dots S_N}\left(s_1,\dots,s_d\right)=C\left(F_{S_1}(s_1),\dots,F_{S_d}(s_d);\vartheta\right), 
\end{align}
in which $C\left(.\right):\left[0,1\right]^d\rightarrow\left[0,1\right]$ denotes the copula function that is a joint CDF of $d$ random vectors on the unit cube $\left[0,1\right]^d$ with uniform marginal distributions, i.e., \cite{ghadi2020copula}
\begin{align}
C\left(u_1,\dots u_d;\vartheta\right)=\Pr\left(U_1\leq u_1,\dots,U_d\leq u_d\right),
\end{align}
where $u_j=F_{S_j}(s_j)$ and $\vartheta$ denotes the copula dependence parameter, which measures the structure of dependency between two arbitrarily correlated RVs. Therefore, by choosing an appropriate copula function, the respective joint CDF can be derived. It has been shown in \cite{ghadi2023gaussian} that although there are different types of copula families that can be used for describing unknown dependence structures, the spatial correlation between two arbitrary ports of a FAS can be accurately modeled using the Gaussian copula, especially when the fluid antenna size is large enough. In particular, the correlation coefficient of Jake's model is accurately approximated by the dependence parameter of the  Gaussian copula for any arbitrarily correlated fading channels.\footnote{This approximation was validated through the rank correlation coefficients such as Spearman's $\rho_\mathrm{s}$ and Kendall's $\tau_\mathrm{k}$ in terms of $\varrho_i$, and comparing the scatterplots of the Gaussian copula and Jake's model, i.e.,  $\varrho_i\approx \eta_i$ \cite{ghadi2023gaussian}.} The multivariate Gaussian copula associated with the correlation matrix $\mathbf{R}$ can be found by
\begin{align}
C_\mathrm{G}\left(u_1,\dots,u_d;\eta\right)=\mathbf{\Phi}_\mathbf{R}\left(\varphi^{-1}\left(u_1\right),\dots,\varphi^{-1}\left(u_d\right);\eta\right),
\end{align}
in which $\varphi^{-1}(u_q)=\sqrt{2}\mathrm{erf}^{-1}\left(2u_q-1\right)$ defines the inverse CDF (quantile function) of the standard normal distribution, with $\mathrm{erf}^{-1}(\cdot)$ denoting the inverse of the error function $\mathrm{erf}(s)=\frac{2}{\sqrt{\pi}}\int_0^s \mathrm{e}^{-t^2}\mathrm{d}t$. The term $\mathbf{\Phi}_\mathbf{R}$ denotes the joint CDF of the multivariate normal distribution with a zero mean vector and a correlation matrix $\mathbf{R}$. Also, the term $\eta\in(-1,1)$ denotes the dependence parameter of the Gaussian copula, which measures and controls the degree of dependence between correlated RVs. Moreover, using the chain rule, the density function $c_\mathrm{G}$ of the Gaussian copula is expressed as \cite{ghadi2023gaussian}
\begin{align}
c_\mathrm{G}\left(u_1,\dots,u_d;\eta\right)=\frac{\exp\left(-\frac{1}{2}\left(\boldsymbol{\varphi}^{-1}\right)^T\left(\mathbf{R}^{-1}-\mathbf{I}\right)\boldsymbol{\varphi}^{-1}\right)}{\sqrt{{\rm det}\left(\mathbf{R}\right)}},\label{eq-copula-g}
\end{align}
where $\boldsymbol{\varphi}^{-1}$  denotes a vector containing the quantile function of the standard normal distribution $\varphi^{-1}\left(u_j\right)$, $\mathrm{det}\left(\mathbf{R}\right)$ is the determinant of $\mathbf{R}$, and $\mathbf{I}$ is the identity matrix.

\begin{theorem}\label{thm-cdf}
The CDF of the channel gain $g_{\mathrm{fas},i}$ at the FAS-equipped NOMA user $i$ under correlated Rayleigh fading is given by \eqref{eq-cdf} (see top of this page), in which $\Upsilon\left(x,y\right)$ denotes the lower incomplete gamma function, 
\begin{figure*}
\begin{align}
F_{\mathrm{fas},i}\left(g_\mathrm{fas}\right)=\Phi_{\mathbf{R}_i}\left(\sqrt 2\mathrm{erf}^{-1}\left(\frac{2}{\Gamma\left(\varkappa_i\right)}\Upsilon\left(\varkappa_i,\frac{g_\mathrm{fas}}{\varpi_i}\right)-1\right),\dots,\sqrt 2\mathrm{erf}^{-1}\left(\frac{2}{\Gamma\left(\varkappa_i\right)}\Upsilon\left(\varkappa_i,\frac{g_\mathrm{fas}}{\varpi_i}\right)-1\right);\eta_i\right)\label{eq-cdf}
\end{align}
\hrulefill
\end{figure*}
%\begin{align}
%	\varkappa=\frac{\overline{a}^2+3\left(\zeta\overline{b}\overline{c}\right)^2}{\overline{a}+\zeta\overline{b}\overline{c}}, \quad \text{and} \quad \varpi=\frac{\left(\overline{a}+\zeta\overline{b}\overline{c}\right)^2}{\overline{a}^2+3\left(\zeta\overline{b}\overline{c}\right)^2},
%\end{align}
\begin{align}
\varkappa_i=\frac{d_{\mathrm{b},i}^{-2\alpha}\overline{a}^2+3\left(\zeta d_{\mathrm{b,t}}^{-\alpha}d_{\mathrm{t},i}^{-\alpha}\overline{b}\overline{c}\right)^2}{d_{\mathrm{b},i}^{-\alpha}\overline{a}+\zeta d_{\mathrm{b,t}}^{-\alpha}d_{\mathrm{t},i}^{-\alpha}\overline{b}\overline{c}}
\end{align}
and
\begin{align} 
\varpi_i=\frac{\left(d_{\mathrm{b},i}^{-\alpha}\overline{a}+\zeta d_{\mathrm{b,t}}^{-\alpha}d_{\mathrm{t},i}^{-\alpha}\overline{b}\overline{c}\right)^2}{d_{\mathrm{b},i}^{-2\alpha}\overline{a}^2+3\left(\zeta d_{\mathrm{b,t}}^{-\alpha}d_{\mathrm{t},i}^{-\alpha}\overline{b}\overline{c}\right)^2},
\end{align}
where $\overline{a}=\mathbb{E}\left[g^{n_i}_{\mathrm{b},i}\right]$, $\overline{b}=\mathbb{E}\left[g_\mathrm{b,t}\right]$, and $\overline{c}=\mathbb{E}\left[g^{n_i}_{\mathrm{t},i}\right]$. Moreover, the term $\mathbf{\Phi}_{\mathbf{R}_i}(\cdot)$ represents the joint CDF of the multivariate normal distribution with zero mean vector and the following correlation matrix
\begin{align}
\mathbf{R}_i= \begin{bmatrix}
1 & \eta_{1,2}^i & \hdots & \eta_{1,N_i}^i\\
\eta_{2,1}^i & 1 & \hdots & \eta_{2,N_i}^i\\
\vdots & \vdots & \ddots& \vdots\\
\eta_{N_i,1}^i & \eta_{N_i,2}^i & \hdots& 1\\
\end{bmatrix}.
\end{align}
\end{theorem}

\begin{proof}
See Appendix \ref{app-thm-cdf}.
\end{proof} 

\begin{theorem}\label{thm-pdf}
The PDF of the channel gain  $g_{\mathrm{fas},i}$ at the FAS-equipped NOMA user $i$ under correlated Rayleigh fading is given by \eqref{eq-pdf} (see top of next page), in which
\begin{figure*}
\begin{align}
f_{\mathrm{fas},i}\left(g_\mathrm{fas}\right)=\frac{\left[\frac{1}{\Gamma\left(\varkappa_i\right)\varpi_i^{\varkappa_i}}g_\mathrm{fas}^{\varkappa_i-1}\mathrm{e}^{-\frac{g_\mathrm{fas}}{\varpi_i}}\right]^{N_i}}{\sqrt{{\rm det}\left(\mathbf{R}_i\right)}}\exp\left(-\frac{1}{2}\left(\boldsymbol{\varphi}^{-1}_{g_{\mathrm{eq},i}^{N_i}}\right)^T\left(\mathbf{R}^{-1}_i-\mathbf{I}\right)\boldsymbol{\varphi}^{-1}_{g_{\mathrm{eq},i}^{N_i}}\right)\label{eq-pdf}
\end{align}
\hrulefill
\end{figure*}
\begin{multline}
\boldsymbol{\varphi}^{-1}_{g_{\mathrm{eq},i}^{N_i}}=\left[\sqrt 2\mathrm{erf}^{-1}\left(\frac{2}{\Gamma\left(\varkappa_i\right)}\Upsilon\left(\varkappa_i,\frac{g_\mathrm{fas}}{\varpi_i}\right)-1\right),\dots,\right.\notag\\
\left.\sqrt 2\mathrm{erf}^{-1}\left(\frac{2}{\Gamma\left(\varkappa_i\right)}\Upsilon\left(\varkappa_i,\frac{g_\mathrm{fas}}{\varpi_i}\right)-1\right)\right].
\end{multline}
\end{theorem}

\begin{proof}
See Appendix \ref{app-thm-pdf}.
\end{proof}

\subsection{Communication Performance}
\subsubsection{OP Analysis}
OP is a key performance metric to assess wireless communication systems, which is defined as the probability that the instantaneous SNR $\gamma$ is below the SNR threshold $\gamma_\mathrm{th}$, i.e., $P_\mathrm{out}=\Pr\left(\gamma\leq\gamma_\mathrm{th}\right)$.

\begin{theorem}\label{thm-op-u1}
The OP of the near FAS-equipped user $\mathrm{u}_\mathrm{n}$ over the considered FAS-aided ISABC is given by \eqref{eq-out-n} (see top of next page), in which
\begin{figure*}
	\begin{align}
		P_{\mathrm{out,u_n}}=&\,\Phi_{\mathbf{R}_{\mathrm{u_n}}}\left(\sqrt 2\mathrm{erf}^{-1}\left(\frac{2}{\Gamma\left(\varkappa_\mathrm{u_n}\right)}\Upsilon\left(\varkappa_\mathrm{u_n},\frac{\tilde{\gamma}_\mathrm{max}}{\varpi_\mathrm{u_n}}\right)-1\right),\dots,\sqrt 2\mathrm{erf}^{-1}\left(\frac{2}{\Gamma\left(\varkappa_\mathrm{u_n}\right)}\Upsilon\left(\varkappa_\mathrm{u_n},\frac{\tilde{\gamma}_\mathrm{max}}{\varpi_\mathrm{u_n}}\right)-1\right);\eta_{\mathrm{u_n}}\right)\label{eq-out-n}
	\end{align}
	\hrulefill
\end{figure*} 
%\begin{figure*}
%\begin{align}
%P_{\mathrm{out,u_n}}=&\,1-\left[1-\Phi_{\mathbf{R}_{\mathrm{u_n}}}\left(\sqrt 2\mathrm{erf}^{-1}\left(\frac{2}{\Gamma\left(\varkappa_\mathrm{u_n}\right)}\Upsilon\left(\varkappa_\mathrm{u_n},\frac{\tilde{\gamma}_\mathrm{sic}}{\varpi_\mathrm{u_n}}\right)-1\right),\dots,\sqrt 2\mathrm{erf}^{-1}\left(\frac{2}{\Gamma\left(\varkappa_\mathrm{u_n}\right)}\Upsilon\left(\varkappa_\mathrm{u_n},\frac{\tilde{\gamma}_\mathrm{sic}}{\varpi_\mathrm{u_n}}\right)-1\right);\eta_{\mathrm{u_n}}\right)\right]\notag\\
%&\times \left[1-\Phi_{\mathbf{R}_{\mathrm{u_n}}}\left(\sqrt 2\mathrm{erf}^{-1}\left(\frac{2}{\Gamma\left(\varkappa_\mathrm{u_n}\right)}\Upsilon\left(\varkappa_\mathrm{u_n},\frac{\tilde{\gamma}_\mathrm{u_n}}{\varpi_\mathrm{u_n}}\right)-1\right),\dots,\sqrt 2\mathrm{erf}^{-1}\left(\frac{2}{\Gamma\left(\varkappa_\mathrm{u_n}\right)}\Upsilon\left(\varkappa_\mathrm{u_n},\frac{\tilde{\gamma}_\mathrm{sic}}{\varpi_\mathrm{u_n}}\right)-1\right);\eta_{\mathrm{u_n}}\right)\right]\label{eq-out-n}
%\end{align}
%\hrulefill
%\end{figure*}
\begin{equation}\label{eq-tildes}
\left\{\begin{aligned}
\tilde{\gamma}_\mathrm{sic}&=\frac{\hat{\gamma}_\mathrm{sic}}{\overline{\gamma}\mu_\mathrm{c}\left(p_\mathrm{u_f}-\hat{\gamma}_\mathrm{sic}p_\mathrm{u_n}\right)},\\
\tilde{\gamma}_\mathrm{u_n}&=\frac{\hat{\gamma}_\mathrm{u_n}}{\overline{\gamma}\mu_\mathrm{c}p_\mathrm{u_n}},
\end{aligned}\right.
\end{equation}
where $\hat{\gamma}_{\mathrm{sic}}$ and $\hat{\gamma}_\mathrm{u_\mathrm{n}}$ are the corresponding SINR thresholds.
\end{theorem} 

\begin{proof}
OP arises for the near user $\mathrm{u}_\mathrm{n}$ when it fails to decode its own signal, the far user's signal, or both. Thus, by definition, the OP for $\mathrm{u_n}$ can be mathematically expressed as
\begin{align}\label{eq-def-op1}
P_\mathrm{out,u_\mathrm{n}} = 1- \Pr\left(\gamma_\mathrm{sic}>\hat{\gamma}_{\mathrm{sic}},\gamma_\mathrm{u_\mathrm{n}}>\hat{\gamma}_\mathrm{u_\mathrm{n}}\right), 
\end{align}
in which $\hat{\gamma}_\mathrm{sic}$ is the SINR threshold of $\gamma_{\mathrm{sic}}$ and $\hat{\gamma}_\mathrm{u_n}$ denotes the SNR threshold of $\gamma_\mathrm{u_n}$. Given that the OP in the NOMA scenario is mainly affected by the power allocation factors, $P_\mathrm{out,u_n}$ becomes $1$ when  $p_\mathrm{u_f}\leq\hat{\gamma}_\mathrm{sic}p_\mathrm{u_n}$. Therefore, for the case that $p_\mathrm{u_f}>\hat{\gamma}_\mathrm{sic}p_\mathrm{u_n}$, the OP can be found as
\begin{align}\notag
	&P_\mathrm{out,u_n}\\
	&\overset{(a)}{=}1- \Pr\left(\frac{\overline{\gamma}p_\mathrm{u_f}\mu_\mathrm{c}g_\mathrm{fas,u_n}}{\overline{\gamma}p_\mathrm{u_n}\mu_\mathrm{c}g_\mathrm{fas,u_n}+1}>\hat{\gamma}_{\mathrm{sic}},\overline{\gamma}p_\mathrm{u_n}\mu_\mathrm{c}g_{\mathrm{fas},\mathrm{u_n}}>\hat{\gamma}_\mathrm{u_n}\right)\\
	&=1-\Pr\left(g_{\mathrm{fas},\mathrm{u_n}}>\frac{\hat{\gamma}_\mathrm{sic}}{\overline{\gamma}\mu_\mathrm{c}\left(p_\mathrm{u_f}-\hat{\gamma}_\mathrm{sic}p_\mathrm{u_n}\right)},g_{\mathrm{fas},\mathrm{u_n}}>\frac{\hat{\gamma}_\mathrm{u_n}}{\overline{\gamma}\mu_\mathrm{c}p_\mathrm{u_n}}\right)\\
	&=1-\Pr\left(g_\mathrm{fas,u_n}>\max\left\{\tilde{\gamma}_\mathrm{sic},\tilde{\gamma}_\mathrm{u_n}\right\}\right)\\
	& = F_{g_{\mathrm{fas},\mathrm{u_n}}}\left(\tilde{\gamma}_\mathrm{max}\right)\label{eq-pr-u1}
\end{align}
in which $(a)$ is achieved by inserting \eqref{eq-sinr-sic} and \eqref{eq-snr-un} into \eqref{eq-def-op1}, $\tilde{\gamma}_\mathrm{sic}=\frac{\hat{\gamma}_\mathrm{sic}}{\overline{\gamma}\mu_\mathrm{c}\left(p_\mathrm{u_f}-\hat{\gamma}_\mathrm{sic}p_\mathrm{u_n}\right)}$, $\tilde{\gamma}_\mathrm{u_n}=\frac{\hat{\gamma}_\mathrm{u_n}}{\overline{\gamma}\mu_\mathrm{c}p_\mathrm{u_n}}$, and  $\tilde{\gamma}_\mathrm{max}=\max\left\{\tilde{\gamma}_\mathrm{sic},\tilde{\gamma}_\mathrm{u_n}\right\}$. Then, by utilizing the CDF of $g_{\mathrm{fas},i}$ from $\eqref{eq-cdf}$, the proof is accomplished.   
%\begin{align}
%&P_\mathrm{out,u_n}\notag\\
%&\overset{(a)}{=}1-\left[1- \Pr\left(\frac{\overline{\gamma}p_\mathrm{u_f}\mu_\mathrm{c}g_\mathrm{fas,u_n}}{\overline{\gamma}p_\mathrm{u_n}\mu_\mathrm{c}g_\mathrm{fas,u_n}+1}\leq\hat{\gamma}_{\mathrm{sic}}\right)\right]\notag\\
%&\quad\quad\times\left[1-\Pr\left(\overline{\gamma}p_\mathrm{u_n}\mu_\mathrm{c}g_{\mathrm{fas},\mathrm{u_n}}\leq\hat{\gamma}_\mathrm{u_n}\right)\right]\\\notag
%&=1-\left[1-\Pr\left(g_{\mathrm{fas},\mathrm{u_n}}\leq\sqrt{\frac{\hat{\gamma}_\mathrm{sic}}{\overline{\gamma}\mu_\mathrm{c}\left(p_\mathrm{u_f}-\hat{\gamma}_\mathrm{sic}p_\mathrm{u_n}\right)}}\right)\right]\\
%&\quad\quad\times\left[1-\Pr\left(g_{\mathrm{fas},\mathrm{u_n}}\leq\sqrt{\frac{\hat{\gamma}_\mathrm{u_n}}{\overline{\gamma}\mu_\mathrm{c}p_\mathrm{u_n}}}\right)\right]\\
%& = 1-\overline{F}_{g_{\mathrm{fas},\mathrm{u_n}}}\left(\frac{\hat{\gamma}_\mathrm{sic}}{\overline{\gamma}\mu_\mathrm{c}\left(p_\mathrm{u_f}-\hat{\gamma}_\mathrm{sic}p_\mathrm{u_n}\right)}\right)\overline{F}_{g_{\mathrm{fas},\mathrm{u_n}}}\left(\frac{\hat{\gamma}_\mathrm{u_n}}{\overline{\gamma}\mu_\mathrm{c}p_\mathrm{u_n}}\right),\label{eq-pr-u1}
%\end{align}
%in which $\overline{F}_{g_\mathrm{fas,u_n}}(r)=1-F{g_\mathrm{fas,u_n}}(r)$ is the complementary CDF (CCDF) of $g_\mathrm{fas,u_n}$. Also, $(a)$ involves inserting \eqref{eq-sinr-sic} and \eqref{eq-snr-un} into \eqref{eq-def-op1} and accounting for the independence of the events. Then, by utilizing the CDF of $g_{\mathrm{fas},i}$ from $\eqref{eq-cdf}$, the proof is accomplished.  
\end{proof}

\begin{theorem}
The OP of the far FAS-equipped user $\mathrm{u_f}$ over the considered FAS-aided ISABC is given by \eqref{eq-op2} (see top of this page), in which 
\begin{figure*}
\begin{align}
P_\mathrm{out,u_f} =\Phi_{\mathbf{R}_\mathrm{u_f}}\left(\sqrt 2\mathrm{erf}^{-1}\left(\frac{2}{\Gamma\left(\varkappa_\mathrm{u_f}\right)}\Upsilon\left(\varkappa_\mathrm{u_f},\frac{\tilde{\gamma}_\mathrm{u_f}}{\varpi_\mathrm{u_f}}\right)-1\right),\dots,\sqrt 2\mathrm{erf}^{-1}\left(\frac{2}{\Gamma\left(\varkappa_\mathrm{u_f}\right)}\Upsilon\left(\varkappa_\mathrm{u_f},\frac{\tilde{\gamma}_\mathrm{u_f}}{\varpi_\mathrm{u_f}}\right)-1\right);\eta_{\mathrm{u_f}}\right)\label{eq-op2}
\end{align}
\hrulefill
\end{figure*}
\begin{align}\label{eq-tilde2}
\tilde{\gamma}_\mathrm{u_f}=\frac{\hat{\gamma}_\mathrm{u_f}}{\overline{\gamma}\mu_\mathrm{c}\left(p_\mathrm{u_f}-\hat{\gamma}_\mathrm{u_f}p_\mathrm{u_n}\right)}. 
\end{align}
\end{theorem}

\begin{proof}
By definition, the OP for the far user $\mathrm{u}_2$ is defined as
\begin{align}
P_\mathrm{out,u_f} = \Pr\left(\gamma_\mathrm{u_f}\leq\hat{\gamma}_\mathrm{u_f}\right),
\end{align}
where $\hat{\gamma}_{\mathrm{u_f}}$ is the SINR threshold of SINR $\gamma_{\mathrm{u_f}}$. Following the same approach as in the proof of Theorem \ref{thm-op-u1}, the OP for $\mathrm{u_f}$ when $p_\mathrm{u_f}>\hat{\gamma}_\mathrm{u_f}p_\mathrm{u_n}$ can be formulated as
\begin{align}
P_\mathrm{out,u_f}
& = \Pr\left(\frac{\overline{\gamma}p_\mathrm{u_f}\mu_\mathrm{c}g_\mathrm{fas,u_f}}{\overline{\gamma}p_\mathrm{u_n}\mu_\mathrm{c}g_\mathrm{fas,u_f}+1}\leq\hat{\gamma}_\mathrm{u_f}\right)\\
&=\Pr\left(g_{\mathrm{fas},\mathrm{u_f}}\leq\frac{\hat{\gamma}_\mathrm{u_f}}{\overline{\gamma}\mu_\mathrm{c}\left(p_\mathrm{u_f}-\hat{\gamma}_\mathrm{u_f}p_\mathrm{u_n}\right)}\right)\\
&=F_{g_{\mathrm{fas},\mathrm{u_f}}}\left(\frac{\hat{\gamma}_\mathrm{u_f}}{\overline{\gamma}\mu_\mathrm{c}\left(p_\mathrm{u_f}-\hat{\gamma}_\mathrm{u_f}p_\mathrm{u_n}\right)}\right). \label{eq-pr-u2}
\end{align}
Now, by considering the CDF of $g_{\mathrm{fas},i}$, \eqref{eq-op2} is obtained and the proof is completed. 
\end{proof}

\begin{corollary}
The asymptotic OP of the FAS-equipped users $\mathrm{u_n}$ and $\mathrm{u_f}$ in the high SNR regime, i.e., $\overline{\gamma}\rightarrow\infty$, is given by \eqref{eq-out-n-asym} and \eqref{eq-op2-asym}, respectively (see top of this page). 
\begin{figure*}
%\begin{align}
%P_{\mathrm{out,u_n}}^\infty=&\,1-\left[1-\Phi_{\mathbf{R}_{\mathrm{u_n}}}\left(\sqrt 2\mathrm{erf}^{-1}\left(\frac{2\tilde{\gamma}_\mathrm{sic}^{\varkappa_\mathrm{u_n}}}{\Gamma\left(\varkappa_\mathrm{u_n}\right)\varkappa_\mathrm{u_n}\varpi_\mathrm{u_n}^{\varkappa_\mathrm{u_n}}}-1\right),\dots,\sqrt 2\mathrm{erf}^{-1}\left(\frac{2\tilde{\gamma}_\mathrm{sic}^{\varkappa_\mathrm{u_n}}}{\Gamma\left(\varkappa_\mathrm{u_n}\right)\varkappa_\mathrm{u_n}\varpi_\mathrm{u_n}^{\varkappa_\mathrm{u_n}}}-1\right);\eta_{\mathrm{u_n}}\right)\right]\notag\\
%&\times \left[1-\Phi_{\mathbf{R}_{\mathrm{u_n}}}\left(\sqrt 2\mathrm{erf}^{-1}\left(\frac{2\tilde{\gamma}_\mathrm{u_n}^{\varkappa_\mathrm{u_n}}}{\Gamma\left(\varkappa_\mathrm{u_n}\right)\varkappa_\mathrm{u_n}\varpi_\mathrm{u_n}^{\varkappa_\mathrm{u_n}}}-1\right),\dots,\sqrt 2\mathrm{erf}^{-1}\left(\frac{2\tilde{\gamma}_\mathrm{u_n}^{\varkappa_\mathrm{u_n}}}{\Gamma\left(\varkappa_\mathrm{u_n}\right)\varkappa_\mathrm{u_n}\varpi_\mathrm{u_n}^{\varkappa_\mathrm{u_n}}}-1\right);\eta_{\mathrm{u_n}}\right)\right]\label{eq-out-n-asym}
%\end{align}
\begin{align}
	P_{\mathrm{out,u_n}}^\infty=&\,\Phi_{\mathbf{R}_{\mathrm{u_n}}}\left(\sqrt 2\mathrm{erf}^{-1}\left(\frac{2\tilde{\gamma}_\mathrm{max}^{\varkappa_\mathrm{u_n}}}{\Gamma\left(\varkappa_\mathrm{u_n}\right)\varkappa_\mathrm{u_n}\varpi_\mathrm{u_n}^{\varkappa_\mathrm{u_n}}}-1\right),\dots,\sqrt 2\mathrm{erf}^{-1}\left(\frac{2\tilde{\gamma}_\mathrm{max}^{\varkappa_\mathrm{u_n}}}{\Gamma\left(\varkappa_\mathrm{u_n}\right)\varkappa_\mathrm{u_n}\varpi_\mathrm{u_n}^{\varkappa_\mathrm{u_n}}}-1\right);\eta_{\mathrm{u_n}}\right)\label{eq-out-n-asym}
\end{align}
\hrulefill
\begin{align}
P_\mathrm{out,u_f}^\infty =\Phi_{\mathbf{R}_\mathrm{u_f}}\left(\sqrt 2\mathrm{erf}^{-1}\left(\frac{2\tilde{\gamma}_\mathrm{u_f}^{\varkappa_\mathrm{u_f}}}{\Gamma\left(\varkappa_\mathrm{u_f}\right)\varkappa_\mathrm{u_f}\varpi_\mathrm{u_f}^{\varkappa_\mathrm{u_f}}}-1\right),\dots,\sqrt 2\mathrm{erf}^{-1}\left(\frac{2\tilde{\gamma}_\mathrm{u_f}^{\varkappa_\mathrm{u_f}}}{\Gamma\left(\varkappa_\mathrm{u_f}\right)\varkappa_\mathrm{u_f}\varpi_\mathrm{u_f}^{\varkappa_\mathrm{u_f}}}-1\right);\eta_{\mathrm{u_f}}\right)\label{eq-op2-asym}
\end{align}
\hrulefill
\end{figure*}
\end{corollary}

\begin{proof}
By applying the relation of $\Upsilon\left(s,x\right)\approx x^s/s$ as $x\rightarrow0$ to \eqref{eq-cdf}, the proof is accomplished. 
\end{proof}

\subsubsection{ECR Analysis}
The ECR of the proposed FAS-aided ISABC system with NOMA user $i$ is defined as
\begin{align}
\overline{\mathcal{C}}_i=\mathbb{E}\left[\log_2\left(1+\gamma_i\right)\right]=\int_0^\infty \log_2\left(1+\gamma_i\right)f_{\gamma_i}\left(\gamma_i\right)\mathrm{d}\gamma_i,\label{eq-def-c}
\end{align}
in which  $f_{\gamma_i}\left(\gamma_i\right)$ is the PDF of $\gamma_i$ that can be derived from \eqref{eq-pdf}. Therefore, by considering this, the ECR for NOMA user $i$ is provided in the following theorem. 

\begin{theorem}
The ECR of FAS-equipped users $\mathrm{u_n}$ and $\mathrm{u_f}$ over the considered FAS-aided ISABC is, respectively, given by
%\begin{align}\label{eq-c-un}
%\overline{\mathcal{C}}_\mathrm{u_n}=\int_0^\infty \frac{\log_2\left(1+\gamma_\mathrm{u_n}\right)}{\overline{\gamma}p_\mathrm{u_n}\mu_\mathrm{c}}f_{g_\mathrm{fas,u_n}}\left(\frac{\gamma_{\mathrm{u_n}}}{\overline{\gamma}p_\mathrm{u_n}\mu_\mathrm{c}}\right)\mathrm{d}\gamma_\mathrm{u_n},
%\end{align}
\begin{align}\label{eq-c-un}
	\overline{\mathcal{C}}_\mathrm{u_n}=\int_0^\infty \log_2\left(1+\overline{\gamma}p_\mathrm{u_n}\mu_\mathrm{c}g_\mathrm{fas,u_n}\right)f_{g_\mathrm{fas,u_n}}\left(g_\mathrm{fas,u_n}\right)\mathrm{d}g_\mathrm{fas,u_n},
\end{align}
and
%\begin{align}
%&\overline{\mathcal{C}}_\mathrm{u_f}=\notag\\
%&\int_0^\infty\frac{\log_2\left(1+\gamma_{\mathrm{u_f}}\right)}{\overline{\gamma}\mu_\mathrm{c}\left(p_\mathrm{u_f}-p_\mathrm{u_n}\gamma_{\mathrm{u_f}}\right)}f_{g_\mathrm{fas,u_f}}\left(\frac{\gamma_{\mathrm{u_f}}}{\overline{\gamma}\mu_\mathrm{c}\left(p_\mathrm{u_f}-p_\mathrm{u_n}\gamma_{\mathrm{u_f}}\right)}\right)\mathrm{d}\gamma_{\mathrm{u_f}},\label{eq-c-uf}
%\end{align}
\begin{align}
	&\overline{\mathcal{C}}_\mathrm{u_f}=\notag\\
	&\int_0^\infty\log_2\left(1+\frac{\overline{\gamma}p_\mathrm{u_f}\mu_\mathrm{c}g_\mathrm{fas,u_f}}{\overline{\gamma}p_\mathrm{u_n}\mu_\mathrm{c}g_\mathrm{fas,u_f}+1}\right)f_{g_\mathrm{fas,u_f}}
	\left(g_\mathrm{fas,u_f}\right)\mathrm{d}g_\mathrm{fas,u_f},\label{eq-c-uf}
\end{align}
where $f_{g_{\mathrm{fas},i}}\left(g_{\mathrm{fas},i}\right)$ is defined in \eqref{eq-pdf}.
\end{theorem}

\begin{proof}
Inserting \eqref{eq-pdf} into \eqref{eq-def-c} completes the proof.
\end{proof}

Computing the integrals in \eqref{eq-c-un} and \eqref{eq-c-uf} is mathematically intractable since $f_{g_{\mathrm{fas},i}}\left(g_{\mathrm{fas},i}\right)$ is defined in terms of the PDF of multivariate and univariate normal distributions as given by \eqref{eq-pdf}, which is quite complicated. Thus, the ECR needs to be solved numerically using multi-paradigm languages such as MATLAB. Additionally, it can be approximated mathematically by exploiting the numerical integration approaches such as the GLQ technique as defined in the following lemma.

\begin{lemma}\label{eq-lemma-glr}
The GLR for a function $\Lambda\left(x\right)$ in the real line domain $\mathbb{R}$, for some $M$, is defined as \cite{abramowitz1972handbook}
\begin{align}
\int_0^\infty \Lambda(x)dx\approx\sum_{n=1}^{M}w_{m}\mathrm{e}^{\epsilon_m}\Lambda\left(\epsilon_m\right),
\end{align}
where $\epsilon_m$ is the $m$-th root of Laguerre polynomial $L_{M}\left(\epsilon_{m}\right)$ and $w_{m}=\frac{\epsilon_{m}}{2\left(M+1\right)^2L^2_{M+1}\left(\epsilon_{m}\right)}$ is the corresponding weight.
\end{lemma} 

\begin{corollary}\label{thm-ecr-un}
The ECR of the near FAS-equipped user $u_\mathrm{n}$ over the considered FAS-aided ISABC is given by \eqref{eq-ecr-un} (see top of this page), in which
\begin{figure*}
\begin{multline}\label{eq-ecr-un}
\overline{\mathcal{C}}_\mathrm{u_n}\approx\sum_{m=1}^{M}\frac{w_{m}\mathrm{e}^{\epsilon_m}\log_2\left(1+\overline{\gamma}p_\mathrm{u_n}\mu_\mathrm{c}\epsilon_m\right)}{\sqrt{{\rm det}\left(\mathbf{R}_\mathrm{u_n}\right)}}\left[\frac{1}{\Gamma\left(\varkappa_\mathrm{u_n}\right)\varpi_\mathrm{u_n}^{\varkappa_\mathrm{u_n}}}\epsilon_m^{\varkappa_\mathrm{u_n}-1}\mathrm{e}^{-\frac{\epsilon_m}{\varpi_\mathrm{u_n}}}\right]^{N_\mathrm{u_n}}\exp\left(-\frac{1}{2}\left(\boldsymbol{\tilde{\varphi}}^{-1}_{g_{\mathrm{eq},\mathrm{u_n}}^{N_\mathrm{u_n}}}\right)^T\left(\mathbf{R}^{-1}_\mathrm{u_n}-\mathbf{I}\right)\boldsymbol{\tilde{\varphi}}^{-1}_{g_{\mathrm{eq},\mathrm{u_n}}^{N_\mathrm{u_n}}}\right)
\end{multline}
\hrulefill
\end{figure*}
\begin{align}\notag
&\boldsymbol{\tilde{\varphi}}^{-1}_{g_{\mathrm{eq},\mathrm{u_n}}^{N_\mathrm{u_n}}}=\Bigg[\sqrt 2\mathrm{erf}^{-1}\left(\frac{2}{\Gamma\left(\varkappa_\mathrm{u_n}\right)}\Upsilon\left(\varkappa_\mathrm{u_n},\frac{\epsilon_m}{\varpi_\mathrm{u_n}}\right)-1\right),\\ \notag
&\, \dots, \sqrt 2\mathrm{erf}^{-1}\left(\frac{2}{\Gamma\left(\varkappa_\mathrm{u_n}\right)}\Upsilon\left(\varkappa_\mathrm{u_n},\frac{\epsilon_m}{\varpi_\mathrm{u_n}}\right)-1\right)\Bigg].
\end{align}
\end{corollary}

\begin{proof}
By applying Lemma \ref{eq-lemma-glr} to \eqref{eq-c-un} and performing some simplifications, the proof is completed.
\end{proof}

\begin{corollary}
The ECR of the far FAS-equipped user $u_\mathrm{f}$ over the considered FAS-aided ISABC is given by \eqref{eq-ecr-uf} (see top of this page), in which
\begin{figure*}
\begin{multline}\label{eq-ecr-uf}
\overline{\mathcal{C}}_\mathrm{u_f}\approx\sum_{m=1}^{M}\frac{w_{m}\mathrm{e}^{\epsilon_m}\log_2\left(1+\frac{\overline{\gamma}p_\mathrm{u_f}\mu_\mathrm{c}\epsilon_m}{\overline{\gamma}p_\mathrm{u_n}\mu_\mathrm{c}\epsilon_m+1}\right)}{\sqrt{{\rm det}\left(\mathbf{R}_\mathrm{u_f}\right)}}
\left[\frac{1}{\Gamma\left(\varkappa_\mathrm{u_f}\right)\varpi_\mathrm{u_f}^{\varkappa_\mathrm{u_f}}}\epsilon_m^{\varkappa_\mathrm{u_f}-1}\mathrm{e}^{-\frac{\epsilon_{m}}{\varpi_\mathrm{u_f}}}\right]^{N_\mathrm{u_f}}\hspace{-1mm}\exp\left(-\frac{1}{2}\left(\boldsymbol{\breve{\varphi}}^{-1}_{g_{\mathrm{eq},\mathrm{u_f}}^{N_\mathrm{u_f}}}\right)^T \left(\mathbf{R}^{-1}_\mathrm{u_f}-\mathbf{I}\right)\boldsymbol{\breve{\varphi}}^{-1}_{g_{\mathrm{eq},\mathrm{u_f}}^{N_\mathrm{u_f}}}\right)
\end{multline}
\hrulefill
\end{figure*}
\begin{align}\notag
	&\boldsymbol{\breve{\varphi}}^{-1}_{g_{\mathrm{eq},\mathrm{u_f}}^{N_\mathrm{u_f}}}=\Bigg[
	\sqrt 2\mathrm{erf}^{-1}\left(\frac{2}{\Gamma\left(\varkappa_\mathrm{u_f}\right)}\Upsilon\left(\varkappa_\mathrm{u_f},\frac{\epsilon_{m}}{\varpi_\mathrm{u_f}}\right)-1\right),\\ \notag
	&\dots, 2\mathrm{erf}^{-1}\left(\frac{2}{\Gamma\left(\varkappa_\mathrm{u_f}\right)}\Upsilon\left(\varkappa_\mathrm{u_f},\frac{\epsilon_{m}}{\varpi_\mathrm{u_f}}\right)-1\right)\Bigg].%\label{eq:ecr-corr3}
\end{align}
\end{corollary}

\begin{proof}
By applying Lemma \ref{eq-lemma-glr} to \eqref{eq-c-uf} and performing some simplifications, the proof is completed.
\end{proof}

\subsection{Sensing Performance}
To evaluate sensing performance, we exploit the concept of estimation SR as defined in \cite{chiriyath2015inner}. From an information-theoretic viewpoint, estimation SR is analogous to the data information rate in wireless communication systems. It represents the estimated rate of target parameter acquisition, where a higher SR indicates better performance in target detection. Therefore, given that the BS utilizes the received signal \eqref{eq-eco} to sense the TSR, the SR can be defined as
\begin{align}
R_\mathrm{t}\leq\frac{\beta}{2T}\log_2\left(1+2T\gamma_\mathrm{echo}\right),
\end{align}
in which $T$ denotes the radar pulse duration, $\beta$ defines the radar's duty cycle, and  $\gamma_\mathrm{echo}$ is the echo SNR, defined as 
\begin{align}
\gamma_\mathrm{echo}=\overline{\gamma}\zeta d_\mathrm{t,b}^{-\alpha} \left|h_{\mathrm{tsr,t}}\right|^2\sigma^2_\mathrm{echo}=\frac{\pi^2}{3}\overline{\gamma}\zeta d_\mathrm{t,b}^{-\alpha} g_{\mathrm{t,b}}g_{\mathrm{b,t}}\sigma^2_{t_\mathrm{df}}, \label{eq-snd-echo}
\end{align}
where $g_{\mathrm{t,b}}=\left|h_{\mathrm{t,b}}\right|^2$ and $g_{\mathrm{b,t}}=\left|h_{\mathrm{b,t}}\right|^2$ are the channel gains from the tag to the BS and from the BS to the tag, respectively. Then we define the ESR to quantify the average  estimation rate. This metric can be considered as the dual counterpart to the data information rate, denoted as
\begin{align}\notag
\overline{R}_\mathrm{t}&\leq\mathbb{E}\left[R_\mathrm{t}\right]\\
&=\frac{\beta}{2T}\int_0^\infty\log_2\left(1+2T\gamma_\mathrm{echo}\right)f_{\gamma_{\mathrm{echo}}}\left(\gamma_{\mathrm{echo}}\right)\mathrm{d}\gamma_\mathrm{echo}.\label{eq-eesr-def}
\end{align}

Note that the SR metric is characterized by the CRLB for estimation, defined as $\sigma^2_{\mathrm{est}}\leq\frac{\sigma^2_{t_\mathrm{df}}}{2T\gamma_\mathrm{echo}}$, where $\sigma^2_\mathrm{est}$  is the time-delay estimation. In this regard, the ESR is first defined as \begin{align}
R_\mathrm{t}\leq\frac{H_{y_\mathrm{b}}-H_\mathrm{est}}{T_\mathrm{bit}},
\end{align}
where $T_\mathrm{bit}=\frac{T}{\beta}$ represents the bits per pulse repetition interval, $H_{y_\mathrm{b}}$ and $H_\mathrm{est}$ denote the entropy of received signal and the entropy of errors, defined as $H_{y_\mathrm{b}}=\frac{1}{2}\log_2\left(2\pi\mathrm{e}\left(\sigma^2_{t_\mathrm{df}}+\sigma^2_\mathrm{est}\right)\right)$ and $H_\mathrm{est}=\frac{1}{2}\left(2\pi\mathrm{e}\sigma^2_\mathrm{est}\right)$, respectively. Therefore, the ESR is greatly affected by the radar's time-delay estimation, and it can be derived by exploiting the CRLB. 

\begin{theorem}\label{thm-eesr}
The ESR of the considered FAS-aided ISABC is given by
%	\begin{align}
%	\overline{R}_\mathrm{t}\leq\frac{\beta}{4 T^2 \overline{b}\overline{e} \ln 2}G_{1,3}^{3,1}\left(\begin{array}{c}
%		\frac{3}{2\pi^2\overline{\gamma}\sigma^2_{t_\mathrm{df}}\overline{b}\overline{e}T}\end{array}
%	\Big\vert\begin{array}{c}
%		-1\\
%		-1,-1,0\\
%	\end{array}\right), \label{eq-eesr}
%	\end{align}
\begin{align}\label{eq-eesr}
\overline{R}_\mathrm{t}\leq\frac{\beta}{2T}\log_2\left(1+\frac{2}{3}T\overline{b}\overline{e}\pi^2\zeta d_\mathrm{t,b}^{-\alpha}\overline{\gamma}\sigma^2_{t_\mathrm{df}}\right),
\end{align}
where $\overline{e}=\mathbb{E}\left[g_\mathrm{t,b}\right]$.
\end{theorem}

\begin{proof}
See Appendix \ref{app-thm-eesr}.
\end{proof}
%\begin{align}
%	P_\mathrm{out,u_n}=\Pr\left(\gamma_{\mathrm{sic}}\leq\tilde{\gamma}_\mathrm{sic},\gamma_{\mathrm{u_n}}\leq\tilde{\gamma}_\mathrm{u_n}\right)
%\end{align}
%\begin{align}
%P_\mathrm{out,u_f}=\Pr\left(\gamma_{\mathrm{u_f}}\leq\tilde{\gamma}_\mathrm{u_f}\right)
%\end{align}
\vspace{0.5cm}
\begin{figure}[!t]
\centering
\includegraphics[width=0.9\columnwidth]{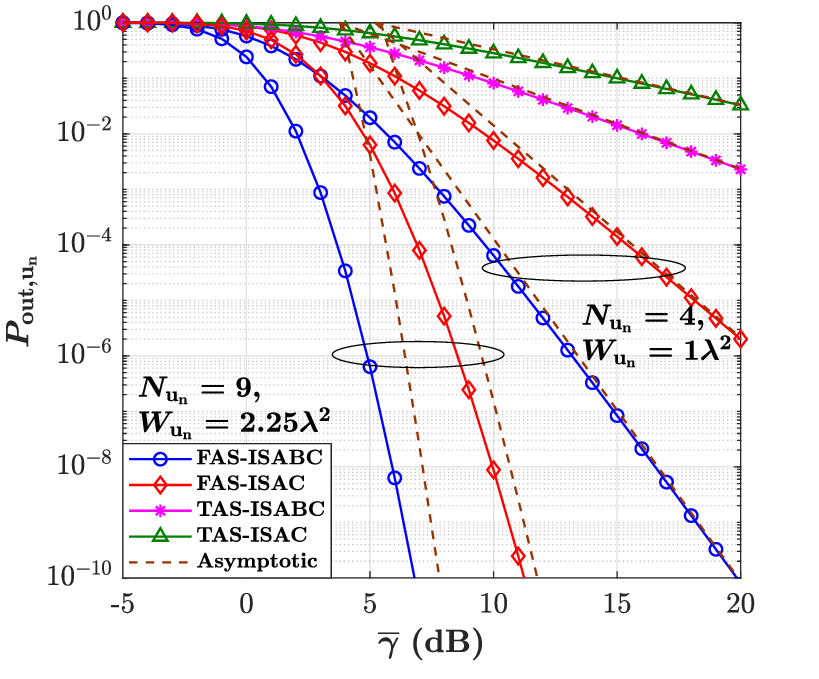}
\caption{OP versus average SNR $\overline{\gamma}$ for the near user $\mathrm{u_n}$, based on \eqref{eq-out-n}.}\label{fig_o_un}
\end{figure}

\begin{figure}[!t]
\centering
\includegraphics[width=0.9\columnwidth]{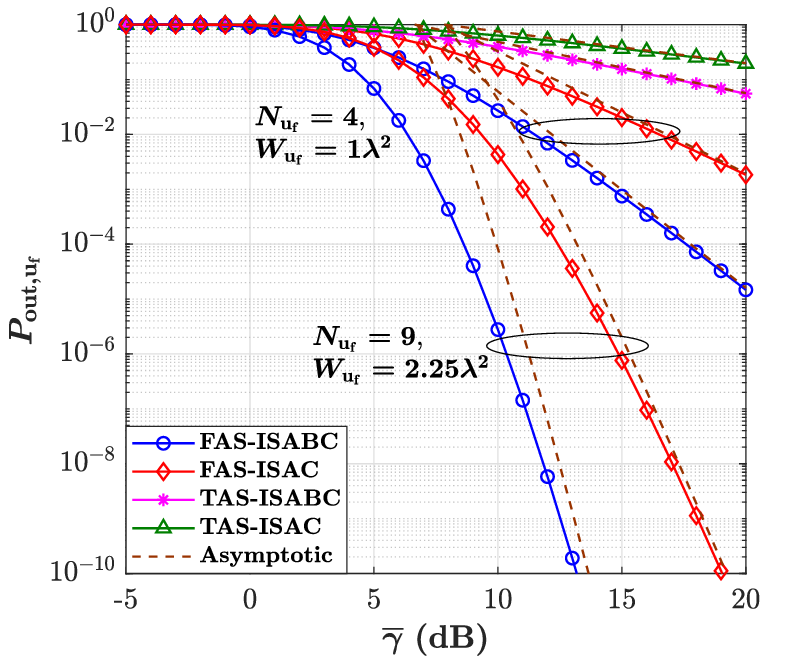}
\caption{OP versus average SNR $\overline{\gamma}$ for the far user $\mathrm{u_f}$, based on \eqref{eq-op2}.}\label{fig_o_uf}
\end{figure}

\begin{figure}[!t]
\centering
\includegraphics[width=0.9\columnwidth]{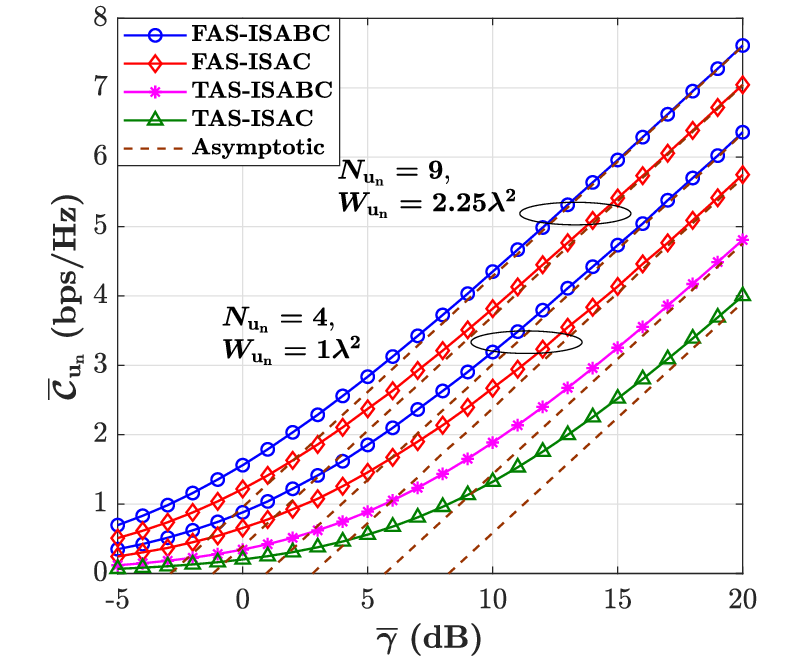}
\caption{ECR versus average SNR $\overline{\gamma}$ for the near user $\mathrm{u_n}$, based on \eqref{eq-ecr-un}.}\label{fig_c_un}
\end{figure}

\begin{figure}[!t]
\centering
\includegraphics[width=0.9\columnwidth]{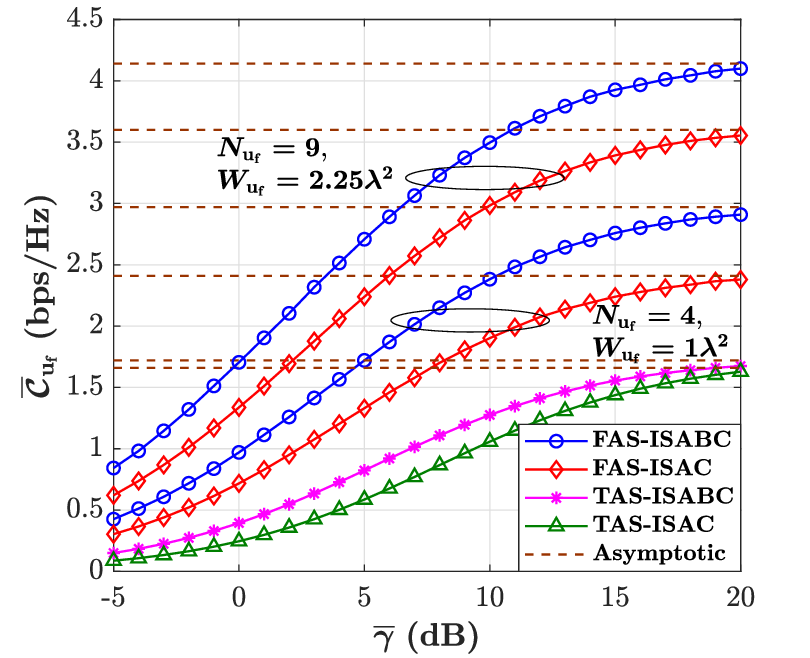}
\caption{ECR versus average SNR $\overline{\gamma}$ for the far user $\mathrm{u_f}$, based on \eqref{eq-ecr-uf}.}\label{fig_c_uf}
\end{figure}

\begin{figure}[!t]
\centering
\includegraphics[width=0.9\columnwidth]{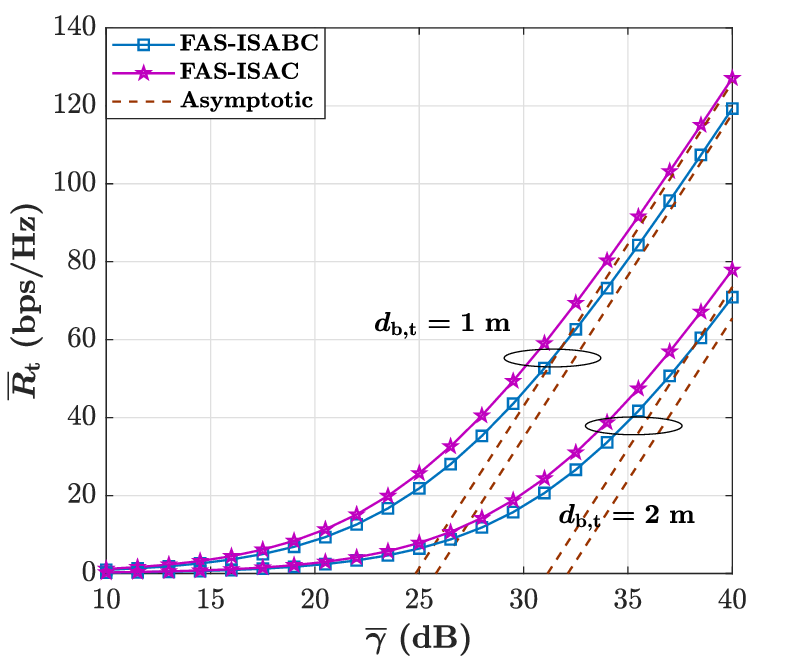}
\caption{ESR versus average SNR $\overline{\gamma}$ for different values of $d_\mathrm{b,t}$, based on \eqref{eq-eesr}.}\label{fig_r_s}
\end{figure}

\begin{figure}[!t]
\centering
\includegraphics[width=0.9\columnwidth]{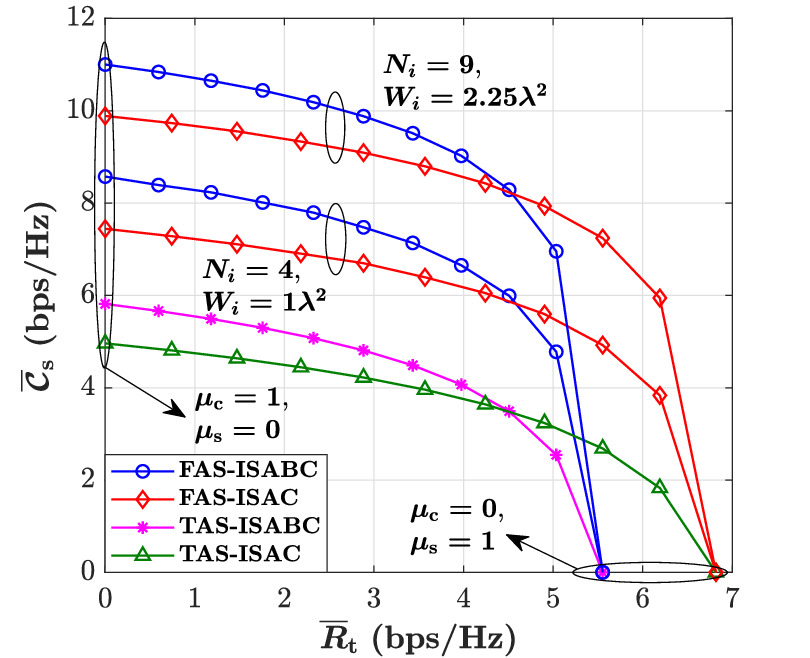}
\caption{ECR versus ESR for different $\mu$ and a given value of $\overline{\gamma}=15$ dB.}\label{fig_c_r}
\end{figure}

\section{Numerical Results}\label{sec-num}
In this section, our derived analytical results are evaluated, which are double-checked in all instances by Monte Carlo simulations. To this end, we set the parameters as $p_\mathrm{u_n}=0.3$, $p_\mathrm{u_f}=0.7$, $\mu_\mathrm{c}=\mu_\mathrm{s}=0.5$, $\zeta=0.8$, $\overline{a}=\overline{b}=\overline{c}=\overline{e}=1$, $\hat{\gamma}_\mathrm{u_f}=\hat{\gamma}_\mathrm{u_n}=\hat{\gamma}_\mathrm{u_{sic}}=0$ dB, $d_\mathrm{b,t}=d_\mathrm{b,u_n}=d_\mathrm{b,u_f}=1$ m, $d_\mathrm{b,u_f}=1.3$ m, $\alpha=2.1$,  $\beta=0.5$, $\sigma^2_{t_\mathrm{df}}=0.1$, and $T=0.01$s. Furthermore, we utilize Algorithm 1 in \cite{ghadi2023gaussian} to simulate the Gaussian copula and its corresponding density functions.

We considered the following benchmarks for comparison:
\begin{itemize}
\item FAS-ISAC: ISAC with FAS-equipped users. This can be theoretically achieved by setting $\zeta=0$ in \eqref{eq-yc} for communication performance without a backscatter tag and by setting $\zeta=1$ in \eqref{eq-eco} for sensing performance with no reflection loss at the tag.
\item TAS-ISABC: ISABC with TAS-equipped users.
\item TAS-ISAC: ISAC with TAS-equipped users.
\end{itemize}

The OP performance againts the average SNR $\overline{\gamma}$ for near and far FAS-equipped users $u_\mathrm{n}$ and $u_\mathrm{f}$ are illustrated in Figs.~\ref{fig_o_un} and \ref{fig_o_uf}, respectively, where the asymptotic results closely match the numerical results at high SNRs. First of all, in both figures, it can be observed that the proposed ISABC model provides a lower OP when the NOMA users are equipped with FAS rather than TAS. This is mainly due to the dynamic reconfigurability feature of FAS, allowing for real-time optimization of the antenna's performance based on the current environment and requirements. In this regard, we can also see that increasing the number of FAS ports $N_i$ and the FAS size $W_i$ significantly improves the OP performance for both NOMA users. Moreover, it is observed that the FAS-ISABC model improves the OP performance compared to ISAC with FAS-equipped users. This indicates that considering the backscatter tag is significantly more beneficial for communication performance than for a conventional target. For instance, for FAS-equipped NOMA users $\mathrm{u_n}$ and $\mathrm{u_f}$ with $N_i=4$ and $W_i=1\lambda^2$ at $\overline{\gamma}=10$dB, the OP under conventional ISAC is on the order of $10^{-2}$ and $10^{-1}$, respectively, while the OP for ISABC-FAS is on the order of $10^{-4}$ and $10^{-2}$. Furthermore, the asymptotic results indicate the accuracy of our theoretical analysis in the high SNR regime for both users.

In Figs.~\ref{fig_c_un} and \ref{fig_c_uf}, the ECR performance versus $\overline{\gamma}$ is depicted for the near and far users $\mathrm{u_n}$ and $\mathrm{u_f}$, respectively. In both figures, the results show that the FAS-ISABC achieves higher ECR values compared to FAS-ISAC. This improvement is expected since the backscatter tag transmits additional data to the NOMA users, which can provide a higher CR for NOMA users. By carefully comparing the curves, we find that for the FAS-equipped users $\mathrm{u_n}$ and $\mathrm{u_f}$ having $N_i=4$ and $W_i=1\lambda^2$ at $\overline{\gamma}=10$ dB, the ECR under ISABC is almost $15\%$ and $20\%$ higher than the ECR under FAS-ISAC for users $\mathrm{u_n}$ and $\mathrm{u_f}$, respectively. Furthermore, the results indicate that the FAS has significant effects on both the ISABC and conventional ISAC scenarios. For instance, under the same FAS configuration setup, the ECR provided by the FAS-ISABC for users $\mathrm{u_n}$ and $\mathrm{u_f}$ is almost $41\%$ and $47\%$ higher than that of TAS-ISABC for $\mathrm{u_n}$ and $\mathrm{u_f}$, respectively. Additionally, we see that the FAS-ISAC provides $50\%$ and $45\%$ higher ECR compared to TAS-ISAC for users $\mathrm{u_n}$ and $\mathrm{u_f}$, respectively. Furthermore, it is worth mentioning that such improvement can be enhanced when higher values of $N_i$ and $W_i$ are considered for the FAS setup or higher $\zeta$ is assumed for the backscatter tag over the communication stage. Moreover, at high SNR, we exploit the approximation of $\log_2\left(1+\textrm{SNR}\right)\approx\log_2\left(\textrm{SNR}\right)$ in \eqref{eq-ecr-un} and \eqref{eq-ecr-uf} to plot the asymptotic behaviour of the ECR. We see that the asymptotic curves reach the ECR at high SNR. 

The ESR performance in terms of $\overline{\gamma}$ is illustrated in Fig.~\ref{fig_r_s}. It can be observed that the ISAC outperforms ISABC, as there is no reflection loss at the target, i.e., $\zeta=1$. Additionally, as expected, the ESR decreases as the backscatter target is located farther away from the BS due to more severe path loss effects. Moreover, by using $\log_2\left(1+\textrm{SNR}\right)\approx\log_2\left(\textrm{SNR}\right)$ for the high SNR regime in \eqref{eq-eesr}, the asymptotic results are provided that shows a perfect match with the ESR. 

Fig.~\ref{fig_c_r} characterizes the performance interplay between the ECR and the ESR for a given average SNR $\overline{\gamma}$. Specifically, by allocating different power coefficients $\mu_\chi$ at the BS for communication and sensing stages, we plot the sum of ECR (i.e., $\overline{\mathcal{C}}_\mathrm{s}=\overline{\mathcal{C}}_\mathrm{u_n}+\overline{\mathcal{C}}_\mathrm{u_f}$) for both NOMA users in terms of the ESR. To do so, we set $\mu_\mathrm{c}=\mu$, and therefore $\mu_\mathrm{s}=1-\mu$, and change the power coefficient $\mu$ from $0$ to $1$. In this regard, by carefully comparing the curves, it can be observed that both ergodic rates (i.e, communication and sensing) achieved by TAS-ISABC and TAS-ISAC are always lower than those achievable by FAS-ISABC and FAS-ISAC, respectively, which confirms the superiority of FAS than TAS.

Furthermore, FAS-ISABC and TAS-ISABC perform better (worse) in terms of the communication (sensing) rate compared with conventional FAS-ISAC and TAS-ISAC, respectively. Therefore, for both the ISABC and ISAC scenarios, this specific  behaviour highly depends on the value of power coefficient at the BS, where larger (smaller) $\mu$ causes  higher (lower) CR and lower (higher) SR. More precisely, the ESR is zero when $\mu=1$ since all the power is utilized for communication stage and no power is considered for target detection. Thus, we reach non-zero ECR with zero ESR. Additionally, by assuming $\mu=0$, the highest value of ESR is achieved, while the ECR cannot be reduced to zero due to the ISAC principle. Hence, it is concluded that without any CR constraints, the SR is maximized with the available transmit power. Conversely, when NOMA users require a non-zero CR, the power coefficient tends to diminish the SR.

\section{Conclusion}\label{sec-con}
In this paper, we studied the performance of a FAS-aided ISABC system under NOMA. Specifically, we considered an ISABC scenario in which a backscatter tag provides sensing information to the BS while simultaneously transmitting additional data to the FAS-equipped NOMA users. For the proposed model, we first derived closed-form expressions for the CDF and PDF of the equivalent channel at both NOMA users using moment matching techniques and the Gaussian copula. Then, we obtained closed-form expressions for the OP and the corresponding asymptotic derivations in the high SNR regime. Additionally, we approximated the ECR using the numerical integration GLQ technique. Moreover, we derived the ESR in closed form by utilizing the CRLB. Eventually, our numerical results provided useful insights into achieving an accurate trade-off between SR and CR. Furthermore, they showed that FAS significantly enhances the overall performance of both conventional ISAC and ISABC compared to TAS. 

\appendices
\section{Proof of Theorem \ref{thm-cdf}}\label{app-thm-cdf}
Given \eqref{eq-yc}, the equivalent channel at the $n_i$-th port of user $i$ is given by $g_{\mathrm{eq},i}^{n_i}=d_{\mathrm{b},i}^{-\alpha}g_{\mathrm{b},i}^{n_i}+\zeta d_{\mathrm{b,t}}^{-\alpha}d_{\mathrm{t},i}^{-\alpha} g_\mathrm{b,t}g^{n_i}_{\mathrm{t},i}$. Without loss of generality and for ease of display, we consider
\begin{align}
g_{\mathrm{eq},i}^{n_i}=d_{\mathrm{b},i}^{-\alpha}A+\zeta d_{\mathrm{b,t}}^{-\alpha}d_{\mathrm{t},i}^{-\alpha}BC=A'+D,
\end{align}
where, given that all channels undergo Rayleigh fading, the channel gains follow exponential distributions. Hence, we have the PDF and CDF for the channel gain $R\in\left\{A,B,C\right\}$ as
\begin{equation}\label{eq-cdf-marg}
\left\{\begin{aligned}
f_R(r)&=\frac{1}{\overline{r}}\exp\left(-\frac{r}{\overline{r}}\right),\\
F_R(r)&=1-\exp\left(-\frac{r}{\overline{r}}\right), 
\end{aligned}\right.
\end{equation}
in which $\overline{r}=\mathbb{E}\left[R\right]$ is the expected value of the RV $R$ for $r\in\left\{a,b,c\right\}$. To derive the distribution of $g_{\mathrm{eq},i}^{n_i}$, we first obtain the distribution RV $D=\zeta d_{\mathrm{b,t}}^{-\alpha}d_{\mathrm{t},i}^{-\alpha}BC$, which  involves the product of two exponentially-distributed RVs. Hence, by using the definition, we have
\begin{align}
F_{D}(d)&\hspace{1mm}=\int_0^\infty f_{B}\left(b\right)F_{C}\left(\frac{d_{\mathrm{b,t}}^{\alpha}d_{\mathrm{t},i}^{\alpha}d}{\zeta b}\right)\mathrm{d}b\\
&\overset{(a')}{=}1-\frac{1}{\overline{b}}\int_0^\infty \exp\left(-\left(\frac{b}{\overline{b}}+\frac{d_{\mathrm{b,t}}^{\alpha}d_{\mathrm{t},i}^{\alpha}d}{\zeta b\overline{c}}\right)\right)\mathrm{d}b\\
&\overset{(b')}{=}1-2\sqrt{\frac{d_{\mathrm{b,t}}^{\alpha}d_{\mathrm{t},i}^{\alpha}d}{\zeta\overline{b}\overline{c}}}\mathcal{K}_1\left(2\sqrt{\frac{d_{\mathrm{b,t}}^{\alpha}d_{\mathrm{t},i}^{\alpha}d}{\zeta\overline{b}\overline{c}}}\right),\label{eq-bessel}
\end{align}
where $(a')$ is obtained by considering the univariate marginal distributions from \eqref{eq-cdf-marg} and $(b')$ is derived by utilizing the integral format provided in \cite[3.471.9]{gradshteyn2007table}. In \eqref{eq-bessel}, $\mathcal{K}_\nu(\cdot)$ represents the $\nu$-order modified Bessel function
of the second kind. Then, by utilizing the PDF definition $f(d)=\frac{\partial }{\partial d}F_D(d)$, the corresponding PDF can be formulated as
\begin{align}\label{eq-pdf-d}
f_{D}(d)=\frac{2d_{\mathrm{b,t}}^{\alpha}d_{\mathrm{t},i}^{\alpha}}{\zeta\overline{b}\overline{c}}\mathcal{K}_0\left(2\sqrt{\frac{d_{\mathrm{b,t}}^{\alpha}d_{\mathrm{t},i}^{\alpha}d}{\zeta\overline{b}\overline{c}}}\right).
\end{align}
Next, the distribution of $g_{\mathrm{eq},i}^{n_i}=A'+D$, which includes the sum of two independent RVs, can be defined as
\begin{align}
F_{g_{\mathrm{eq},i}^{n_i}}\left(g_\mathrm{eq}\right)=\int_0^\infty F_{A'}\left(g_\mathrm{eq}-d\right)f_{D}\left(d\right)\mathrm{d}d \label{eq-app1}. 
\end{align}
However, deriving an exact statistical characterization of the CDF in \eqref{eq-app1} seems challenging. To overcome this issue, we resort to the moment matching technique to approximate the corresponding CDF \cite{al2010approximation}. In this context, the Gamma distribution is often utilized to approximate complex distributions due to its simplicity, characterized only by two adjustable parameters: the shape parameter $\varkappa$ and the scale parameter $\varpi$. Thus, the mean and variance of such approximation Gamma distribution are defined as $\varkappa\varpi$ and $\varkappa\varpi^2$, respectively. 

Now, let us attempt to find the CDF of $g_{\mathrm{eq},i}^{n_i}$ using the moment matching technique. By  definition, the mean and variance of $D$ can be calculated as $\mathbb{E}\left[D\right]=\zeta d_{\mathrm{b,t}}^{-\alpha}d_{\mathrm{t},i}^{-\alpha} \overline{b}\overline{c}$ and $\mathrm{Var}\left(D\right)=3\left(\zeta d_{\mathrm{b,t}}^{-\alpha}d_{\mathrm{t},i}^{-\alpha}\overline{b}\overline{c}\right)^2$, respectively. Similarly, the mean and variance of $A'$ are given by $\mathbb{E}\left[A'\right]=d_{\mathrm{b},i}^{-\alpha}\overline{a}$ and $\mathrm{Var}\left(A'\right)=d_{\mathrm{b},i}^{-2\alpha}\overline{a}^2$, respectively. Now, by matching the mean and variance of the RV $g_{\mathrm{eq},i}^{n_i}=A'+D$ with the mean $\varkappa_i\varpi_i$ and variance $\varkappa_i\varpi_i^2$ of the Gamma distribution, we have 
\begin{align}\label{eq-cdf-eq}
F_{g_{\mathrm{eq},i}^{n_i}}\left(g_\mathrm{eq}\right)\approx\frac{1}{\Gamma\left(\varkappa_i\right)}\Upsilon\left(\varkappa_i,\frac{g_\mathrm{eq}}{\varpi_i}\right),
\end{align}
in which 
\begin{align}
\varkappa_i=\frac{d_{\mathrm{b},i}^{-2\alpha}\overline{a}^2+3\left(\zeta d_{\mathrm{b,t}}^{-\alpha}d_{\mathrm{t},i}^{-\alpha}\overline{b}\overline{c}\right)^2}{d_{\mathrm{b},i}^{-\alpha}\overline{a}+\zeta d_{\mathrm{b,t}}^{-\alpha}d_{\mathrm{t},i}^{-\alpha}\overline{b}\overline{c}}
\end{align}
and
\begin{align} \varpi_i=\frac{\left(d_{\mathrm{b},i}^{-\alpha}\overline{a}+\zeta d_{\mathrm{b,t}}^{-\alpha}d_{\mathrm{t},i}^{-\alpha}\overline{b}\overline{c}\right)^2}{d_{\mathrm{b},i}^{-2\alpha}\overline{a}^2+3\left(\zeta d_{\mathrm{b,t}}^{-\alpha}d_{\mathrm{t},i}^{-\alpha}\overline{b}\overline{c}\right)^2}.
\end{align}
Furthermore, by extending the joint CDF definition, the CDF of $g_{\mathrm{fas},i}$ can be mathematically expressed as
\begin{align}
&F_{\mathrm{fas},i}\left(g_\mathrm{fas}\right)\notag\\
&=\Pr\left(\max\left\{g_{\mathrm{eq},i}^1,\dots,g_{\mathrm{eq},i}^{N_i}\right\}\leq g_\mathrm{fas}\right)\\
&=\Pr\left(g_{\mathrm{eq},i}^1\leq g_\mathrm{fas},\dots,g_{\mathrm{eq},i}^{N_i}\leq g_\mathrm{fas}\right)\\
&=F_{g_{\mathrm{eq},i}^1,\dots,g_{\mathrm{eq},i}^{N_i}}\left(g_\mathrm{fas},\dots,g_\mathrm{fas}\right)\\
&\hspace{-0.5ex}\overset{(c')}{=}C\left(F_{g_{\mathrm{eq},i}^1}\left(g_\mathrm{fas}\right),\dots,F_{g_{\mathrm{eq},i}^{N_i}}\left(g_\mathrm{fas}\right);\vartheta_i\right)\\
&\hspace{-0.5ex}\overset{(d')}{=}\mathbf{\Phi}_{\mathbf{R}_i}\left(\varphi^{-1}\left(F_{g_{\mathrm{eq},i}^1}\left(g_\mathrm{fas}\right)\right),\dots,\varphi^{-1}\left(F_{g_{\mathrm{eq},i}^{N_i}}\left(g_\mathrm{fas}\right)\right);\eta_i\right),\label{eq-app3}
\end{align}
where $(c')$ is obtained by using Sklar's theorem from \eqref{eq-sklar}, and $(d')$ is determined by considering the definition of Gaussian copula, where $\varphi^{-1}(\cdot)$ is defined as
\begin{align}
	\varphi^{-1}\left(F_{g_{\mathrm{eq},i}^{n_i}}\left(g_\mathrm{fas}\right)\right)=\sqrt{2}\mathrm{erf}^{-1}\left(2F_{g_{\mathrm{eq},i}^{n_i}}\left(g_\mathrm{fas}\right)-1\right).\label{eq-app4}
\end{align}
Eventually, by substituting \eqref{eq-cdf-eq} into \eqref{eq-app4}, and then inserting the results into \eqref{eq-app3}, we derive \eqref{eq-cdf} and complete the proof.

\section{Proof of Theorem \ref{thm-pdf}}\label{app-thm-pdf}
By applying the chain rule, the PDF of $g_\mathrm{fas}$ can be expressed by the product of the marginal PDFs and the Gaussian copula density function \cite{ghadi2023gaussian} as
\begin{align}\notag
&f_{\mathrm{fas},i}\left(g_\mathrm{fas}\right)\\
&= c_\mathrm{G}\left(F_{g_{\mathrm{eq},i}^1}\left(g_\mathrm{fas}\right),\dots,F_{g_{\mathrm{eq},i}^{N_i}}\left(g_\mathrm{fas}\right);\eta\right)\prod_{n_i=1}^{N_i} f_{g_{\mathrm{eq},i}^{n_i}}\left(g_\mathrm{fas}\right).\label{eq-app6}
\end{align}
Hence, we need to derive the PDF of $g_{\mathrm{eq},i}^{n_i}$. By applying the PDF definition $f_X(x)=\frac{\partial}{\partial X}F_X(x)$ to \eqref{eq-cdf-eq}, we have
\begin{align}
f_{g_{\mathrm{eq},i}^{n_i}}\left(g_\mathrm{eq}\right)=\frac{1}{\Gamma\left(\varkappa\right)\varpi^\varkappa}g_\mathrm{eq}^{\varkappa-1}\mathrm{e}^{-\frac{g_\mathrm{eq}}{\varpi}}.\label{eq-app5}
\end{align}
Now, by substituting \eqref{eq-copula-g}, \eqref{eq-cdf-eq}, and \eqref{eq-app5} into \eqref{eq-app6}, we get the desired result and the proof is accomplished.

\section{Proof of Theorem \ref{thm-eesr}}\label{app-thm-eesr}
In order to derive $\overline{R}_\mathrm{t}$, we need to solve the integral in \eqref{eq-eesr-def}, which  mathematically defines the ESR. To accomplish this, we apply  Jensen's inequality so that \eqref{eq-eesr-def} can be rewritten as
\begin{align}
\overline{R}\leq\frac{\beta}{2T}\log_2\left(1+2T\mathbb{E}\left[\gamma_\mathrm{echo}\right]\right), \label{eq-app7}
\end{align}
where $\mathbb{E}\left[\gamma_{\mathrm{echo}}\right]$ is the expected value of $\gamma_{\mathrm{echo}}$.  Next, we need to find the PDF of $\gamma_\mathrm{echo}$. In \eqref{eq-snd-echo}, since  $\gamma_\mathrm{echo}$ involves the product of two exponentially-distributed channel gains $g_\mathrm{b,t}$ and $g_\mathrm{t,b}$, the PDF of $\gamma_\mathrm{echo}$ can be directly derived from \eqref{eq-pdf-d}, i.e., 
\begin{align}
f_{\gamma_\mathrm{echo}}\left(\gamma_\mathrm{echo}\right)=\frac{6d_\mathrm{t,b}^{\alpha}}{\overline{b}\overline{e}\pi^2\overline{\gamma}\zeta\sigma^2_{t_\mathrm{df}}}\mathcal{K}_0\left(2\sqrt{\frac{3\gamma_{\mathrm{echo}}d_\mathrm{t,b}^{\alpha}}{\overline{b}\overline{e}\pi^2\overline{\gamma}\zeta\sigma^2_{t_\mathrm{df}}}}\right),\label{eq-echo}
\end{align}
where $\overline{e}=\mathbb{E}\left[g_\mathrm{t,b}\right]$. Hence, $\mathbb{E}\left[\gamma_{\mathrm{echo}}\right]$ is given by
\begin{align}
\mathbb{E}\left[\gamma_{\mathrm{echo}}\right]&=\int_0^\infty \frac{6\gamma_{\mathrm{echo}}d_\mathrm{t,b}^{\alpha}}{\overline{b}\overline{e}\pi^2\overline{\gamma}\zeta\sigma^2_{t_\mathrm{df}}}\mathcal{K}_0\left(2\sqrt{\frac{3\gamma_{\mathrm{echo}}d_\mathrm{t,b}^{\alpha}}{\overline{b}\overline{e}\pi^2\overline{\gamma}\zeta\sigma^2_{t_\mathrm{df}}}}\right)\mathrm{d}\gamma_{\mathrm{echo}}\\
&=\frac{1}{3}\overline{b}\overline{e}\pi^2\overline{\gamma}\zeta d_\mathrm{t,b}^{-\alpha}\sigma^2_{t_\mathrm{df}}. \label{eq-app8}
\end{align}
Now, by substituting \eqref{eq-app8} into \eqref{eq-app7}, \eqref{eq-eesr} is achieved.

%In order to derive $\overline{R}_\mathrm{t}$, we need to solve the integral in \eqref{eq-eesr-def}, which is mathematically defined the EESR. To do so, we should first find the PDF of $\gamma_\mathrm{echo}$. In \eqref{eq-snd-echo}, given that $\gamma_\mathrm{echo}$ includes the product of two exponentially-distributed channel gains $g_\mathrm{b,t}$ and $g_\mathrm{t,b}$, the PDF of $\gamma_\mathrm{echo}$ can be directly derived from \eqref{eq-pdf-d}, i.e., 
%\begin{align}
%f_{\gamma_\mathrm{echo}}\left(\gamma_\mathrm{echo}\right)=\frac{6}{\overline{b}\overline{e}\pi^2\overline{\gamma}\sigma^2_{t_\mathrm{df}}}\mathcal{K}_0\left(2\sqrt{\frac{3\gamma_{\mathrm{echo}}}{\overline{b}\overline{e}\pi^2\overline{\gamma}\sigma^2_{t_\mathrm{df}}}}\right),\label{eq-echo}
%\end{align}
%where $\overline{e}=\mathbb{E}\left[g_\mathrm{t,b}\right]$. Next, by inserting \eqref{eq-echo} into \eqref{eq-eesr-def}, we have
%\begin{align}\notag
%&\overline{R}_\mathrm{t}\leq \frac{3\beta}{T\overline{b}\overline{e}\pi^2\overline{\gamma}\sigma^2_{t_\mathrm{df}}}\\
%&\times\int_0^\infty \mathcal{K}_0\left(2\sqrt{\frac{3\gamma_{\mathrm{echo}}}{\overline{b}\overline{e}\pi^2\overline{\gamma}\sigma^2_{t_\mathrm{df}}}}\right)\log_2\left(1+2T\gamma_{\mathrm{echo}}\right)\mathrm{d}\gamma_{\mathrm{echo}}.
%\end{align}
%Finally, by solving the above integral with the help of ?, \eqref{eq-eesr} is achieved and the proof is completed. 
\bibliographystyle{IEEEtran}
%\bibliography{sample.bib}


\end{document}